\documentclass[review,onefignum,onetabnum]{siamart190516}

\usepackage{lipsum}
\usepackage{amsfonts}
\usepackage{graphicx}
\usepackage{epstopdf}
\usepackage{algorithmic}
\ifpdf
  \DeclareGraphicsExtensions{.eps,.pdf,.png,.jpg}
\else
  \DeclareGraphicsExtensions{.eps}
\fi

\newsiamremark{remark}{Remark}
\newsiamremark{hypothesis}{Hypothesis}
\crefname{hypothesis}{Hypothesis}{Hypotheses}
\newsiamthm{claim}{Claim}

\title{Geometrical Characterization of Sensor Placement for Cone-Invariant and Multi-Agent Systems against Undetectable Zero-Dynamics Attacks\thanks{Submitted to the editors DATE.
\funding{This research is supported in part by City University of Hong Kong under the projects 7004866, 7005065, in part by the Knut and Alice Wallenberg Foundation, in part by the Swedish Foundation for Strategic Research, and in part by the Research
Grants Council of Hong Kong Special Administrative Region, China, under the Theme-Based Research Scheme T23-701/14-N.}}}

\author{Jianqi Chen\thanks{Department of Electrical Engineering, City University of Hong Kong, Hong Kong, China
  (\email{jianqchen2-c@my.cityu.edu.hk},  \email{jichen@cityu.edu.hk} ).}
\and Jieqiang Wei \thanks{ACCESS Linnaeus Centre, School of Electrical Engineering, KTH Royal Institute of Technology, Sweden
  (\email{jieqiang@kth.se}, \email{hsan@kth.se},  \email{kallej@kth.se}).}
\and Wei Chen \thanks{Department of Mechanics and Engineering Science, College of Engineering, Peking University, China
(\email{w.chen@pku.edu.cn}).}  
\and Henrik Sandberg\footnotemark[3]
\and Karl H. Johansson\footnotemark[3]
\and Jie Chen\footnotemark[2]
}

\usepackage{amsopn}

\usepackage{graphicx}      
\usepackage{amsmath}
\usepackage{amssymb}
\usepackage{subcaption}
\usepackage{verbatim}
\usepackage{dsfont}
\usepackage{xcolor}
\usepackage{upgreek}
\usepackage{amsfonts}
\usepackage{textcomp}
\usepackage[normalem]{ulem}
\usepackage{arydshln,leftidx,mathtools}
\usepackage{caption}
\usepackage{blindtext}
\usepackage{float}
\usepackage{tikz}
\usepackage{enumerate}
\ifpdf
\hypersetup{
  pdftitle={Sensor Placement of Cone-Invariant and Multi-Agent Systems against Undetectable Zero-Dynamics Attacks},
  pdfauthor={Jianqi Chen,  Jieqiang Wei,  Wei Chen,  Henrik Sandberg,  Karl H. Johansson,  and Jie Chen}
}
\fi

\newtheorem{assumption}{Assumption}
\theoremstyle{definition}
\newtheorem{exmp}{Example}[section]

\DeclareMathOperator{\Tr}{Tr}

\DeclareMathOperator{\rank}{\textbf{R}}

\begin{document}
\nolinenumbers
\maketitle

\begin{abstract}
Undetectable attacks are an important class of malicious attacks threatening the 
security  of cyber-physical systems, which can modify a system's state but leave the system output measurements unaffected, and hence cannot be detected from the output.
This paper studies undetectable attacks on cone-invariant systems and multi-agent systems. 
We first provide a general characterization of zero-dynamics attacks, which characterizes fully undetectable attacks targeting the non-minimum phase zeros of a system. This geometrical characterization makes it possible to develop a defense strategy seeking to place a minimal number of sensors to detect and counter the zero-dynamics attacks on the system's actuators.
The detect and defense scheme amounts to computing a set containing potentially vulnerable actuator locations and nodes, and a defense union for feasible placement of sensors based on the geometrical properties of the cones under consideration.

\end{abstract}



\section{Introduction}
Cyber-physical systems (CPS),  networked and enabled by today's ubiquitous information technology (IT)
infrastructure, have been widely regarded as new-generation engineered systems integrated by physical plants, control and communication networks, and computational units. The technological, economic,  and societal impact of such systems is vastly 
expanding due to rapidly increasing network interconnections of
different components,
such as sensors, actuators, control processing units,  and communication devices. Much to our chagrin, however, interconnected network and computing devices expose the vulnerability of CPS and open the door to potential cyber threats.  Malicious attackers can gain access to sensing and 
actuating devices to launch attacks through IT components, which is likely to compromise the safe and reliable 
operation of a CPS and,  in an extreme scenario, lead to catastrophic consequences \cite{cardenas2008research}.

In light of the ever-increasing presence and ever-expanding scope of CPS,
 there have been growing needs to address challenges in
 detecting attack threats,
mitigating the impact of attacks, and designing effective defense strategies.
Of particular relevance to this paper are a class of attacks that can be broadly referred to as  undetectable attacks.
In \cite{liu2011false},   \textit{false data 
injection attacks} were considered for static systems, which inject
false,  ill-attempted data into the system. 
\emph{Stealthy deception attacks},   which alter sensor readings to avoid detection and cause damage,
were studied in \cite{teixeira2010cyber}. 
\textit{Replay attacks} \cite{mo2009secure}
record the system's operating data
and fake the past data as the current operating signal to drive
the system. 
\textit{Zero-dynamics attacks} target the
transmission zeros of the system and hide in the output \cite{pasqualetti2013attack}. 
\textit{Covert attacks} exploit 
decoupling structures to deceive the controller by interrupting the input and the output simultaneously \cite{smith2011decoupled}.

While as alluded to above various undetectable attacks have been designed,
 the work  on the detection and prevention of such  attacks has been scarce. 
Generally, the detection of an attack can be done by inserting additional signals so as to perturb the system's operating data, or by adding additional sensors, so as to change the system's structure and dynamics.
For example,  it was shown in  \cite{mo2009secure} that 
adding a Gaussian signal unknown to the attacker into communication
channels can make \textit{replay attacks} detectable,  and in \cite{teixeira2012revealing},  that
by placing sensors on all neighbors of the potential attacked nodes,  
\textit{zero-dynamics attacks} can be prevented.
Nevertheless,
one should note that adding a large  number of  sensors may neither be effective nor feasible
in many situations, especially in a distributed setting,  because of environmental constraints
and cost-efficient considerations.  As such,  where possible,
more efficient defense strategies with a reasonable number of sensors are called for.

In this paper we study zero-dynamics attacks and their detection,  in both a single-loop cone-invariant system and multi-agent system configuration.
We provide a general mathematical view of such attacks on cone-invariant linear time-invariant (LTI) systems.
Zero-dynamics attacks have been previously investigated  in
\cite{pasqualetti2013attack,teixeira2012revealing,chen2016dynamic,teixeira2015secure,milovsevic2018security}.
These  attacks typically attack the system actuators while hidden from the sensors of the system.   As an example,  when the system contains a real non-minimum phase zero $s_0$,  a zero-dynamics attack can be designed in the form \cite{pasqualetti2013attack} of
$$d(t)=-e^{s_0 t}d_0.$$
This attack signal can alter the system's state response but is unobserved from the system output.  
In this paper,  we present a  
general zero-dynamics attack configuration,  where we show that more generally, non-minimum phase complex zeros can also be subjected to zero attacks.

Our more significant developments are devoted to
the detection and prevention of zero-dynamics attacks.
Via a geometrical framework, we show that for LTI cone-invariant systems,  zero-dynamics attacks can be detected, or a system can be made resistive to such attacks, by developing
 an efficient sensor placement strategy.
 While seemingly a mathematical notion,  cone-invariant systems encompass many systems of practical interest. One particular instance corresponds to positive systems,  whose states form an invariant positive cone. 
Positive systems have many applications, e.g.,
 in modeling growth behaviors of economical systems, ecological systems, population dynamics, and more generally, dynamic systems involving positivity constraints \cite{berman1994nonnegative,shafai1997explicit}. 
 Recently, positive systems have been used to model power grids,
traffic flow, communication/computation networks,  as well as
production planning and logistics \cite{rantzer2015scalable}. 
Another case of interest is 
Lyapunov and Riccati differential equations,  which constitute a positive semi-definite  cone in the space of Hermitian matrices \cite{tanaka2013dc}. 
Finally,  cone-invariance with respect to polyhedral cones and  ellipsoidal cones are also of interest and have been considered in \cite{dorea1995design,grussler2014modified}.

In this paper,   for cone-invariant systems and  multi-agent systems \cite{ren2007information,yu2010some},
we develop a unified defense framework against zero-dynamics  attacks. 
This is accomplished by characterizing 
a vulnerable set containing potential positions or nodes where  attacks may get into the system,
 and meanwhile by constructing  a defensive union including feasible sensor placements  based on the
geometric  properties of the cones under consideration.
Our contributions  can be summarized as follows.  In Section \cref{sec:2},  we provide a mathematical
 background on convex cones and systems that are cone-invariant.
In Section \cref{sec:3},   we provide a general characterization of zero-dynamics attacks on LTI systems,  which enables us to fully characterize all possible attacks utilizing the knowledge of a system's non-minimum phase zeros. Based on this characterization, we develop in Section  \cref{sec:4} defense strategies for cone-invariant systems, which amount to placing a minimal number of sensors to counter zero-dynamics attacks on the system's actuators.
 Section \cref{sec:5} then extends the results to multi-agent systems,
 in which undetectable attacks are launched on certain nodes of the multi-agent systems and the defense is implemented at other nodes.
 Section \cref{sec:6}  discusses generalizations to multi-input,  multi-output  (MIMO) systems that are subject to multiple attacks. 
 Section \cref{sec:7} presents numerical studies on a microgrid power network under attack, to illustrate our defense strategies.  This paper concludes in  Section \cref{sec:8}.


\section{Mathematical Preliminaries} \label{sec:2}
\subsection{Convex Cones}
\begin{definition} 
A nonempty set $\mathcal{K} \subset\mathbb{R}^{n}$ is said to be a
convex cone if $\mathcal{K}+\mathcal{K}\subset\mathcal{K}$ and 
$a\mathcal{K}\subset\mathcal{K}$ for all $a\geq 0$.
\end{definition}

We define the set $-\mathcal{K}=\lbrace x : -x\in \mathcal{K}\rbrace$. 
A convex cone $\mathcal{K}$ is said to be pointed if $\mathcal{K}\cap -\mathcal{K}=\left\lbrace 0 \right\rbrace $ and solid
if $\mathrm{int}\, \mathcal{K}$, i.e., the interior of $\mathcal{K}$, is nonempty.
The boundary of $\mathcal{K}$ is denoted as $\partial\mathcal{K}$. Moreover, a closed,  pointed,   and solid convex cone is called a proper cone.
A proper cone induces a partial order in $\mathbb{R}^n$, that is,  $\forall x,y\in \mathbb{R}^n$, $y\preceq x$ if and only if $x-y \in \mathcal{K}$. 
Hereafter, unless otherwise specified, all cones we mention will be   proper cones.

Herein, we briefly survey some typical and representative cones  \cite{berman1994nonnegative,stern1991exponential,boyd2004convex}.
\begin{itemize}
\item A \textit{polyhedral cone} represents the intersection of finitely many closed half spaces, each containing the origin on its boundary.
Every polyhedral cone intersected by $k$ closed half spaces can be described mathematically as
$$\mathcal{K}=\left\lbrace {x\in \mathbb{R}^{n}:Ax\geq_{+} 0, A\in \mathbb{R}^{k\times n}} \right\rbrace,$$
where $\geq_{+}$ represents the element-wise order.
When $k=n$ and $A=I_n$, then $\mathcal{K}=\mathbb{R}_{+}^n$ denotes the nonnegative orthant, which is also called  \textit{positive cone}. 
 
\item A \textit{positive semi-definite (\textit{PSD})  cone} is one that consists of positive semi-definite matrices, i.e., 
$$\mathcal{S}^n_{+}=\left\lbrace X\in \mathcal{S}^n: v^{\top}Xv\geq 0,\forall v\in \mathbb{R}^{n} \right\rbrace, $$ 
where
$\mathcal{S}^n\in \mathbb{R}^{n\times n}$ denotes the set of $n\times n$ real symmetric matrices. A PSD cone $\mathcal{S}^n_{+}$ can be equivalently described in the $n(n+1)/2$-dimensional
vector space, consisting of the elements of all positive semi-definite
matrices on the  main diagonal and above it.
\item An \textit{ellipsoidal cone} is defined as
$$\mathcal{K}=\left\lbrace x \in \mathbb{R}^{n}:x^{\top}Qx\leq 0, l_{n}^{\top}x\geq 0 \right\rbrace,$$
where $Q\in \mathcal{S}^n$ has inertia\footnote{The inertia of a real symmetric matrix is a triple of the number of positive eigenvalues, zero eigenvalues, and negative eigenvalues.} $(n-1,0,1)$ and $l_n$ is a unit eigenvector of $Q$ corresponding to the negative eigenvalue.
When $Q=Q_n=\mathrm{diag}\left\lbrace -1,I_{n-1} \right\rbrace $ and $l_n=e_1$, the first Euclidean coordinate, then $\mathcal{K}$ becomes a \textit{Lorentz cone }(also called ice-cream cone or second-order cone),  denoted as
$$\mathcal{L}_{+}^n=\left\lbrace x\in \mathbb{R}^n: x_1\geq \sqrt{x_2^2+\cdots +x_n^2} \right\rbrace.$$
\end{itemize}
Illustrations of the positive cone  $\mathbb{R}_{+}^3$, the PSD cone $\mathcal{S}_+^2$ with
$$ \begin{bmatrix}
x_1&x_2\\x_2&x_3
\end{bmatrix}\in \mathcal{S}_+^2,
$$ and the Lorentz cone $\mathcal{L}_{+}^3$ are given in Fig. \ref{fig:three cones}.
\begin{figure}[h!]
\begin{center}
\hspace*{-0.5cm}
\vspace{-0.2cm}
\includegraphics[width=10cm, height=8.0cm]{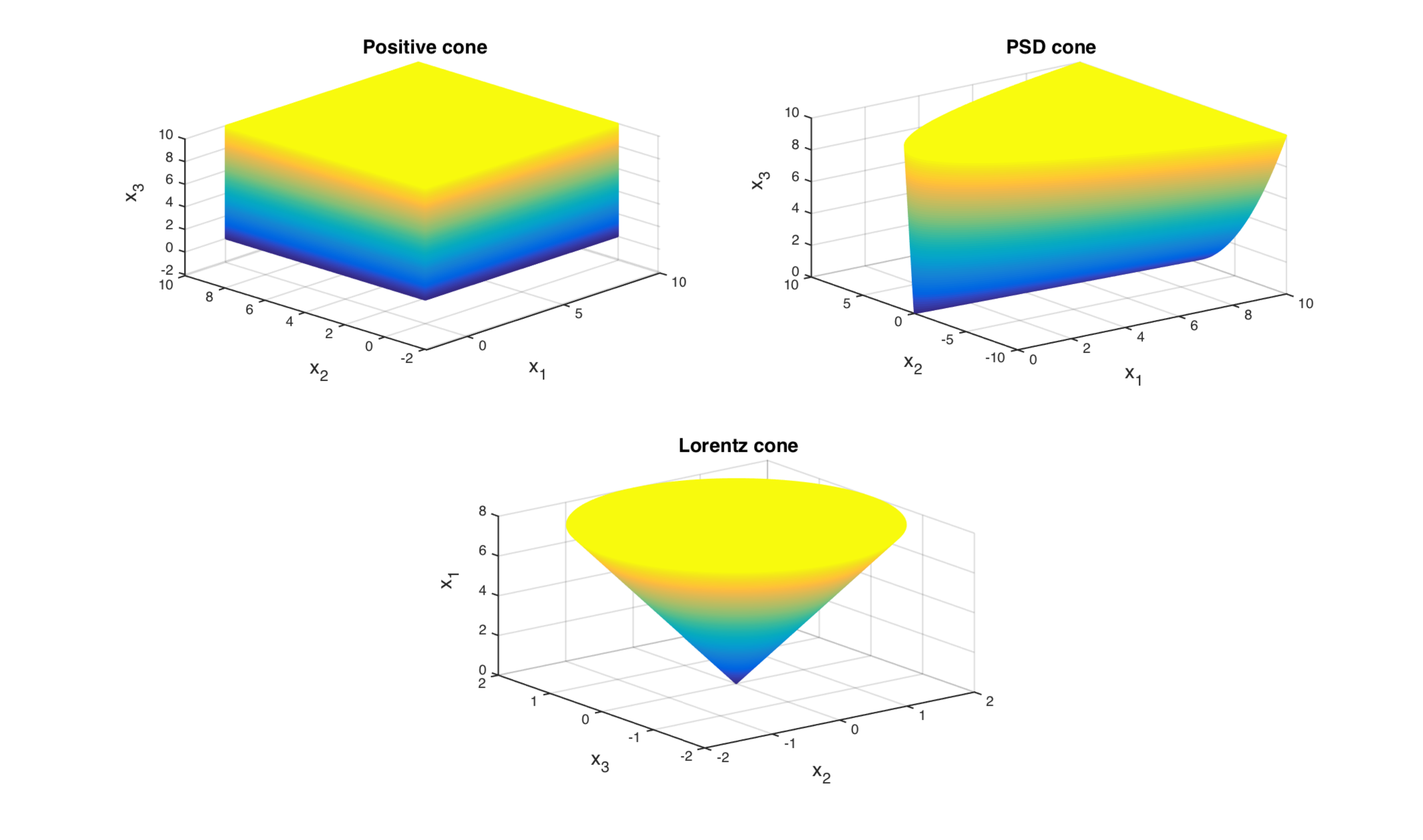}
\end{center}
\caption{Positive cone, PSD cone and Lorentz cone.}
\label{fig:three cones}
\end{figure}

Given a cone $\mathcal{K}$ in $\mathbb{R}^{n}$, its dual cone is defined as
\begin{center}
$\mathcal{K}^{\ast}=\left\lbrace y\in \mathbb{R}^{n}: 
\left\langle x,y \right\rangle \geq 0, \forall x\in \mathcal{K}\right\rbrace $,
\end{center}  
where $\left\langle\, ,\, \right\rangle $ denotes a pre-defined inner product.
A cone $\mathcal{K}$ is said to be self-dual if $\mathcal{K}=\mathcal{K}^{\ast}$. For instance,
$\mathbb{R}_{+}^n$ and $\mathcal{L}_{+}^n$ are self-dual cones endowed with the standard inner product $\left\langle x,y \right\rangle  = x^{\top}y$, and $\mathcal{S}_+^n$ is a self-dual cone endowed with the inner product $\left\langle X,Y \right\rangle =\Tr{(XY)}$. One useful lemma reveals the geometric relationship between each entry in $\mathcal{K}$ and its corresponding $\mathcal{K}^{\ast}$ \cite{berman1994nonnegative}.
\begin{lemma}\label{lemma:cone and its dual}
For a cone $\mathcal{K}$ and its dual cone $\mathcal{K}^{\ast}$, 
it holds that $\left\langle y,x\right\rangle>0 $, $\forall\, y \in \mathrm{int}\,\mathcal{K}^{\ast},\,x \in \mathcal{K}$.
\end{lemma}
\subsection{Matrices}
In the sequel, matrices characterizing the invariance with respect to given cones will be studied.

\begin{definition}\label{definition:cone-invariant}
For a cone $\mathcal{K}$, a matrix $A\in\mathbb{R}^{n \times n}$ is called a $\mathcal{K}$-invariant matrix if $A\mathcal{K}\subset \mathcal{K}$. 
\end{definition}

Denote by $\pi(\mathcal{K})$ the set of cone-invariant matrices over $\mathcal{K}$.
For example,  any nonnegative matrix $A\in \mathbb{R}^{n\times n}$ is an $\mathbb{R}_{+}^n$-invariant matrix. Any matrix $A\!=\!B\otimes B,\,\forall B\in \mathbb{R}^{n\times n}$, is an $\mathcal{S}_{+}^n$-invariant matrix, where $\otimes$ denotes the Kronecker product. The matrix $A\in \mathbb{R}^{n\times n}$ is an $\mathcal{L}_{+}^n$-invariant matrix if it satisfies the relation
\begin{equation}
a_1^{\top}\in \mathcal{L}_{+}^n,\quad A^{\top}Q_{n}A-\delta Q_n\preceq 0,
\nonumber
\end{equation}
for some $\delta\in \mathbb{R}$, where $a_1$ is the first row of $A$ and $\preceq$ is defined on $\mathcal{S}_{+}^n$ \cite{loewy1975positive}.

\begin{definition} \label{definition:K-positive}
A matrix $A\in \pi(\mathcal{K})$ is said to be $\mathcal{K}$-positive  if $A(\mathcal{K} \setminus \left\lbrace 0\right\rbrace )\subset \mathrm{int} \mathcal{K}$.  
\end{definition}
Denote by $\pi^{+}(\mathcal{K})$  the set of $\mathcal{K}$-positive matrices.

Except for cone-invariant matrices, \textit{cross-positive matrices} (exponential cone-invariant matrices) are also introduced herein \cite{schneider1970cross}.
\begin{definition}\label{definition:Cross-positive}
For a cone $\mathcal{K}$ and its dual cone $\mathcal{K}^{\ast}$, a matrix $A\in\mathbb{R}^{n \times n}$ is called cross-positive over $\mathcal{K}$ if 
for all $x\in \mathcal{K}$, $y\in \mathcal{K}^{\ast}$ such that $\left\langle y,x \right\rangle=0$,  it holds that $\left\langle y,Ax \right\rangle\geq 0 $.
\end{definition}

Denote by $\varpi(\mathcal{K})$ the set of cross-positive matrices over $\mathcal{K}$. It is straightforward that $\pi(\mathcal{K})\subset \varpi(\mathcal{K})$ for a given cone $\mathcal{K}$. For instance, $\varpi(\mathbb{R}_{+}^{n})$ is the set of Metzler matrices \cite{berman1994nonnegative}.  Any matrix $A=B\oplus B,\,\forall B\in \mathbb{R}^{n\times n}$, is cross-positive over $\mathcal{S}_{+}^n$, where $\oplus$ denotes the Kronecker sum.
The matrix $A\in \mathbb{R}^{n\times n}$ is cross-positive over $\mathcal{L}_{+}^n$ if 
\begin{equation}
A^{\top}Q_n+Q_nA-\xi Q_n \preceq 0
\nonumber
\end{equation}
for some $\xi\in \mathbb{R}$ \cite{stern1991exponential}, where $\preceq$ is defined on $\mathcal{S}_{+}^n$.

Different from the conventional irreducibility\footnote{The matrix $A$ is irreducible if no permutation matrix $P$ exists such that $P^{\top}AP=\begin{bmatrix}
A_{11}&A_{12}\\0&A_{22}
\end{bmatrix}$, where $A_{11}$ and $A_{22}$ are square submatrices.} of a matrix, the cone-irreducibility of cross-positive matrices is defined below \cite{berman1994nonnegative,schneider1970cross}.

\begin{definition} \label{definition:irreducible matrix}
A matrix $A\in \varpi(\mathcal{K})$ is said to be $\mathcal{K}$-irreducible if no eigenvector of $A$ lies on the boundary of $\mathcal{K}$.
\end{definition}
Denote by $\varpi'(\mathcal{K})$  the set of $\mathcal{K}$-irreducible cross-positive matrices.

\begin{remark} \label{remark:positive irreducible and irreducible matrix}
From \cite{vandergraft1968spectral}, for any matrix $A\in \varpi(\mathbb{R}_{+}^n)$,  the $\mathbb{R}_{+}^n$-irreducibility and the  conventional matrix irreducibility are equivalent.
\end{remark}

Moreover, when $A\in \pi(\mathcal{K})$, there exists another equivalent definition of
 matrix cone-irreducibility in terms of the \textit{face} of the cone.

\begin{definition}\label{definition:face}
Let $\mathcal{F}\subseteq \mathcal{K}$ be a pointed closed cone. Then $\mathcal{F}$ is called a face of $\mathcal{K}$ if $\forall x\in \mathcal{F}$,
\begin{equation}
0\preceq y\preceq x\Rightarrow y\in \mathcal{F},
\nonumber
\end{equation}
where $\preceq$ is defined on $\mathcal{K}$. 
The face $\mathcal{F}$ is said to be nontrivial if $\mathcal{F}\neq \lbrace 0\rbrace$ and $\mathcal{F}\neq \mathcal{K}$.
\end{definition} 

\begin{definition}\label{definition:irreducible cone-invariant matrix}
A matrix $A\in \pi(\mathcal{K})$ is said to be $\mathcal{K}$-irreducible if the only faces of $\mathcal{K}$ leaving invariant with $A$ are $\lbrace 0\rbrace$ or $\mathcal{K}$ itself.
\end{definition}
Denote by $\pi'(\mathcal{K})$ the set of $\mathcal{K}$-irreducible cone-invariant matrices. 
Then the following relations are known 
\cite{schneider1970cross}.
\begin{lemma}\label{lemma:containment}
$$\pi^{+}(\mathcal{K})\subseteq \pi'(\mathcal{K})\subseteq \pi(\mathcal{K}).$$
\end{lemma}

The following lemma,  adapted from \cite{berman1994nonnegative}, on
cross-positive matrices will be useful in the sequel.
\begin{lemma}\label{lemma:eigenvalue and eigenvector of cross-positive}
Let $\mathrm{spectrum}(A)$ be the set of eigenvalues of $A$.

\noindent If $A\in \varpi(\mathcal{K})$, then 

\noindent 1) $\mu = \max \lbrace \mathrm{Re}\, (\delta) : \delta\in\mathrm{spectrum}(A)\rbrace$ is an eigenvalue of $A$, 

\noindent 2) $\mathcal{K}$ and $\mathcal{K}^{\ast}$ contain right- and left-eigenvectors of $A$ associated with the eigenvalue $\mu$.

\noindent If $A\in \varpi'(\mathcal{K})$, then 

\noindent 3) $\mu = \max \lbrace \mathrm{Re}\, (\delta) : \delta\in\mathrm{spectrum}(A)\rbrace$ is a simple eigenvalue of $A$, 

\noindent 4) $\mathrm{int}\,\mathcal{K}$ and $\mathrm{int}\,\mathcal{K}^{\ast}$
contain right- and left-eigenvectors of $A$ associated with the eigenvalue $\mu$.
\end{lemma}

Note that Lemma \ref{lemma:eigenvalue and eigenvector of cross-positive} reduces to the Krein-Rutman theorem if $A\in \pi(\mathcal{K})$,  and furthermore, to the well-known Perron-Frobenius theorem if 
$A\in \pi(\mathbb{R}_{+}^n)$ \cite{berman1994nonnegative}.
\subsection{Cone-Invariant Systems}
Consider a linear time-invariant (LTI)  system with input $u(t)\in \mathbb{R}$, output $y(t)\in \mathbb{R}$, and state $x(t)\in \mathbb{R}^n$,  with realization
\begin{equation}
\begin{aligned}
\dot{x}(t)&=Ax(t)+bu(t),\\
y(t)&=c^{\top}x(t).
\end{aligned}
\label{eq:cis}
\end{equation}
We say that this 
system  is  cone-invariant over a  cone $\mathcal{K}$ if $x(t)\in \mathcal{K}$ and $y(t)\geq 0$ whenever $x(0)\in \mathcal{K}$ and $u(t)\geq 0$ for $t\geq 0$. Equivalently, 
the following lemma reveals the relations of the system's cone-invariance to the triple $(A,b,c^{\top})$ \cite{zheng2016projected,tarbouriech1994positively}.

\begin{lemma}
The LTI system \eqref{eq:cis} is a cone-invariant system over $\mathcal{K}$ if and only if
\begin{center}
 $A\in \varpi(\mathcal{K})$, $b\in \mathcal{K},$  and $c\in \mathcal{K}^{\ast}$.
\end{center}
\label{lemma:CIS}
\end{lemma}

When $\mathcal{K}=\mathbb{R}_{+}^n$, $\mathcal{K}=\mathcal{S}_{+}^n$, or $\mathcal{K}=\mathcal{L}_{+}^n$,  the related cone-invariant system  reduces to a positive system, a dynamic covariance system, or a Lorentz cone-invariant system, respectively, which all possess rich theoretical features and broad applications \cite{rantzer2015scalable,tanaka2013dc,bhattacharya2009cone}. 

In the remainder of the paper,  we study cone-invariant systems over a general cone $\mathcal{K}$.
We examine the behavior of such systems under undetectable attacks
and construct accordingly  a
 defense mechanism to counter the attacks.  We consider
cone-invariant systems that are
Hurwitz stable; in other words, the systems are stable
in the absence of attacks.

\section{Undetectable Attacks} \label{sec:3}
Consider a MIMO LTI system
\begin{equation}
\begin{aligned}
\dot{x}(t)&=Ax(t)+Bd(t),\\
y(t)&=Cx(t),
\end{aligned}
\label{eq:mimolti}
\end{equation}
with output measurement $y(t)\in \mathbb{R}^l$ excited by the initial state $x(0)\in\mathbb{R}^n$ and a malicious attack $d(t)\in \mathbb{R}^m$. 
Throughout this paper,  from \eqref{eq:mimolti},  we address the scenario where the attacks are launched on the system's actuators.
We denote the system output as $y(t)=y(x(0),d(t),t).$
An  intuitive interpretation of undetectable attacks is that such an attack will result in the same output as that
in the absence of attacks 
\cite{pasqualetti2013attack}, and hence cannot be detected from output measurement. 

\begin{definition}\label{definition:undetect attack}
An attack $d(t)\in \mathbb{R}^m$  on the  system \eqref{eq:mimolti} is undetectable if the system output satisfies the relation
\begin{equation}
y(x(0)- \tilde{x}(0),\, d(t),\, t) = 0,\,\forall\, t\geq 0,
\label{eq:undetect attack 1}
\end{equation}
for some $x(0)$ and $\tilde{x}(0)$. 
\end{definition}
Here $x(0)$ and $\tilde{x}(0)$ are understood to be the system's actual initial state and an attack initial condition.
Hence,  the design of an undetectable attack amounts to finding a input $d(t)$ that,   when coordinating with the initial state $x(0)- \tilde{x}(0)$,  yields a zero output.


\begin{lemma}
\label{lemma:real undetectable attack}
For the system \eqref{eq:mimolti},  the attack signal 
\begin{equation}\label{eq:undetectable attack 1}
d(t)=-e^{s_0t}d_0,~s_0\in \mathbb{R},~d_0\in \mathbb{R}^m
\end{equation}
is undetectable if and only if there exists
a vector $\zeta\neq 0$ such that
\begin{equation}\label{eq:rosenbrock}
\begin{bmatrix}
s_0I-A&B\\C&0
\end{bmatrix}\begin{bmatrix}
\zeta\\d_0
\end{bmatrix}=0.
\end{equation}
\end{lemma}
Here the attack initial condition can be obtained by $\tilde{x}(0)=x(0)-\zeta$.
It is rather clear that the attack \eqref{eq:undetectable attack 1}  is
undetectable if and only if $s_0$ is a real transmission zero of the system, and hence is also referred to as a zero-dynamics attack.
Then a natural question is how to construct a real attack signal cooperated with a real attack initial condition to attack the system with complex transmission zeros?
This is answered in the following theorem.

\begin{theorem}\label{theorem:nonzero imginary part}
For the system \eqref{eq:mimolti},    the attack signal
\begin{equation}
d(t)=-l_1{\rm Re}\left(e^{s_0t}d_0 \right)-l_2{\rm Im}\left( e^{s_0t}d_0 \right), ~s_0\in \mathbb{C},~d_0\in\mathbb{C}^m,~\forall\, l_1,\,l_2\in \mathbb{R},
\label{eq:undetectable attack 2}
\end{equation}
coordinated with
\begin{equation}
\tilde{x}(0)=x(0)-\left( l_1{\rm Re}\left( \zeta \right) +l_2{\rm Im} \left( \zeta\right)\right)
\end{equation}
is undetectable if and only if 
 there exists
a vector $\zeta\neq 0$ such that 
 the equation \eqref{eq:rosenbrock} is satisfied.
\end{theorem}
Note that for a zero $s_0 \in \mathbb{R}$,  the undetectable attack \eqref{eq:undetectable attack 2}  reduces to that in \eqref{eq:undetectable attack 1}.   
\begin{proof}
Denote $d_0=d_1+d_2j$, $s_0=\alpha+\beta j$ and $\zeta=x_1+x_2j$. The attack signal \eqref{eq:undetectable attack 2}  can thus be rewritten as 
$$
\begin{aligned}
d(t)&=e^{\alpha t}( sin(\beta t)l_1d_2-\cos(\beta t)l_1d_1-sin(\beta t)l_2d_1-cos(\beta t)l_2d_2)\\
&=e^{\alpha t}( sin(\beta t)(l_1d_2-l_2d_1)-\cos(\beta t)(l_1d_1+l_2d_2)).
\end{aligned}$$
For the initial state,  it holds that $x(0)-\tilde{x}(0)=l_1x_1+l_2x_2$.
Also,   the real and imaginary parts of the equation \eqref{eq:rosenbrock}  are obtained as
\begin{equation}\label{eq:sI-A}
\begin{aligned}
(\alpha I-A) x_1-\beta x_2+Bd_1&=0,
\\
(\alpha I-A) x_2+\beta x_1+Bd_2&=0,
\end{aligned}
\end{equation}
where
$Cx_1=0$,  and $Cx_2=0.$
Denote the Laplace transforms of $d(t)$ and $y(t)$ by $D(s)$ and $Y(s)$.  Therefore,
\begin{equation}
D(s)=\dfrac{\beta}{(s-\alpha)^{2}+\beta^{2}}(l_1d_2-l_2d_1)-\dfrac{s-\alpha}{(s-\alpha)^{2}+\beta ^{2}}(l_1d_1+l_2d_2),
\nonumber
\end{equation}
and
\begin{equation}
\begin{aligned}
Y(s)&=C(sI-A)^{-1}BD(s)+C(sI-A)^{-1}\left( x(0)-\tilde{x}(0)\right) 
\\
&=C(sI-A)^{-1}B\dfrac{(\beta l_1-(s-\alpha)l_2)}{(s-\alpha)^{2}+\beta^{2}}d_2 
-C(sI-A)^{-1}B\dfrac{(\beta l_2+(s-\alpha)l_1)}{(s-\alpha)^{2}+\beta^{2}}d_1
\\
&~~~~~~+C(sI-A)^{-1}(l_1x_1+l_2x_2).
\end{aligned}
\nonumber
\end{equation}
In addition, together with  \eqref{eq:sI-A},  it follows that
\begin{equation}
\begin{aligned}
Y(s)&=-C(sI-A)^{-1}\dfrac{(\beta l_1-(s-\alpha)l_2)}{(s-\alpha)^{2}+\beta^{2}}((\alpha I-A) x_2+\beta x_1)
\\
&~~~~~~~~~~+C(sI-A)^{-1}\dfrac{(\beta l_2+(s-\alpha)l_1)}{(s-\alpha)^{2}+\beta^{2}}((\alpha I-A) x_1-\beta x_2)
\\
&~~~~~~~~~~+C(sI-A)^{-1}\dfrac{(s-\alpha)^{2}+\beta^{2}}{(s-\alpha)^{2}+\beta^{2}}(l_1x_1+l_2x_2)
\\
&=C(sI-A)^{-1} \dfrac{((s-\alpha)l_1+\beta l_2)}{(s-\alpha )^{2}+\beta ^{2}}(sI-A)x_1 
\\
&~~~~~~~~~~+C(sI-A)^{-1}\dfrac{((s-\alpha)l_2-\beta l_1)}{(s-\alpha)^{2}+\beta^{2}}(sI-A)x_2 
\\
&=\dfrac{((s-\alpha)l_1+\beta l_2)}{(s-\alpha)^{2}+\beta ^{2}}Cx_1 
+\dfrac{((s-\alpha)l_2-\beta l_1) }{(s-\alpha)^{2}+\beta^{2}}Cx_2
=0.
\end{aligned}
\nonumber
\end{equation}
We conclude that $y(x(0)-\tilde{x}(0),d(t),t)=0$,  and the proof is completed.
\end{proof}

We conclude this section by commenting on zero-dynamics attacks on cone-invariant systems. Note that when 
restricting  general LTI systems to  cone-invariant systems over a cone $\mathcal{K}$,   because of the  cone-invariance additionally required,  
the definition of undetectable attacks,   in comparison to Definition \ref{definition:undetect attack},  should be modified accordingly as follows. Specifically,  we say  that
an attack $d(t)\in \mathbb{R}$ is undetectable on the cone-invariant system \eqref{eq:cis} over a cone $\mathcal{K}$ if  Definition \ref{definition:undetect attack} is satisfied, 
and additionally the  state in  \eqref{eq:cis} satisfies the invariance condition
$x(t)\in \mathcal{K},\,\forall\, t\geq 0.$
In other words, for cone-invariant systems,  not only the attack cannot be observed at the output of the system,  but also the cone-invariance property of the system must be preserved.   Clearly,  both the attacks 
$d(t)$  in  \eqref{eq:undetectable attack 1}  and \eqref{eq:undetectable attack 2}  are unobserved at the output.  To ensure that the condition $x(t)\in \mathcal{K}$ is valid,   it is necessary that $d(t)\geq 0,\forall t\geq 0$,  which
can be satisfied for the attack $d(t)$  in  \eqref{eq:undetectable attack 1}  by letting $d_0< 0$.  On the other hand,  the attack $d(t)$ in \eqref{eq:undetectable attack 2} cannot be kept nonnegative  since it is a linear combination of sinusoidal signals.   Consequently, 
for a stable cone-invariant system \eqref{eq:cis},  only the attack $d(t)$ in \eqref{eq:undetectable attack 1}  is still undetectable.

\section{Sensor Placement}  \label{sec:4}
In light of the discussion at the end of section \ref{sec:3},  in this section,  we develop a detection strategy  to cope with  undetectable attacks $d(t)=-e^{s_0t}d_0,$  on the non-minimum phase zeros $s_0\geq 0$ of stable cone-invariant systems.

\subsection{Protection via Sensor Placement} 
We consider a single zero-dynamics attack,  that is, the attack
enters the system  through one actuator or one node. 
 Here without loss of generality,   we take $b \in \left\lbrace e_i,i=1,\ldots,n\right\rbrace $ to represent  the accessibility of the attack to the system,  where $e_i$ is the $i$th Euclidean coordinate.
To characterize the available attacked positions in cone-invariant systems, we make the following assumption.

\begin{assumption}\label{assumption:the non-empty attack set}
The set $\Lambda_{\mathcal{K}}=\left\lbrace e_i: e_i\in \mathcal{K} \right\rbrace $ is nonempty.
\end{assumption}
The set $\Lambda_{\mathcal{K}}$ characterizes all accessible positions for potential attackers.
The cardinality of $\Lambda_{\mathcal{K}}$ varies with the choice of $\mathcal{K}$. For example, 
\begin{equation}\label{eq:Lambda cone}
\begin{aligned}
&\qquad\qquad\qquad  \Lambda_{\mathbb{R}_{+}^n}= \left\lbrace e_i,i=1,\ldots,n \right\rbrace,\\
&\Lambda_{\mathcal{S}_{+}^n}=\left\lbrace \mathrm{vec}(\mathrm{diag}(e_i)),i=1,\ldots,n \right\rbrace,~ \mathrm{and}~\Lambda_{\mathcal{L}_{+}^n}= \lbrace e_1\rbrace.
\end{aligned}
\end{equation}
Clearly,   for these cones,  $\Lambda_{\mathcal{K}}$ is not empty. Analogously, we also assume that the set $\left\lbrace e_j: e_j\in \mathcal{K}^{\ast} \right\rbrace $ is nonempty.

In the rest of this paper,  we seek  defense approaches by designing $c$ to prevent undetectable attacks.  In general,   
adding a large number of sensors 
 implies an effective defensive strategy.
More precisely, by the PBH test \cite{zhou1996robust}, the system's observability implies that $x(0)=\tilde{x}(0)$,  and
hence $bd(t)=0,\forall\, t\geq 0$,  that is,  the undetectable attack is non-existent.
Nevertheless,  this approach is clearly not a cost-efficient way,  especially 
for large-scale and distributed systems.  
Recalling from Lemma \ref{lemma:real undetectable attack} the condition for the attack
 $d(t)=-e^{s_0t}d_0$,   
it is clear that 
a defensive strategy is successful if it can block all non-minimum phase zeros of  the given system.

\begin{definition}\label{definition:successful defense strategy}
We say that a defense strategy, namely a design of $c$,  is (almost) successful if one of the following conditions holds: 
\end{definition}
\begin{enumerate}[(i)]
\item $\rank\left( \begin{bmatrix}
s_0I-A&b\\c^{\top}&0
\end{bmatrix}\right)=n+1 $,
\item $\rank\left( \begin{bmatrix}
s_0I-A&b\\c^{\top}&0
\end{bmatrix}\right)< n+1$ \textit{and} $d_0=0$,
\end{enumerate}
\textit{for each $s_0$ with ($s_0>0$) $s_0\geq 0$.  Here $\textbf{R}$   denotes the rank of a matrix. }

For $s_0\geq 0$,  the condition (i) indicates that by designing an appropriate $c$,
 $s_0$ is not a non-minimum phase zero of system, while the condition (ii), via PBH test,  means that $s_0$ as a non-minimum phase zero is also an unobservable mode of the system. 
For $s_0=0$  with $d_0\neq 0$,  
the corresponding attack $d(t)=d_0$ as a constant signal cannot make the states of a stable system divergent.  Hence we call a strategy is almost successful if the above conditions can be satisfied only for $s_0>0$. 

\subsection{Sensor Placement Strategy}
In this section, we concentrate on the case with a single measurement, i.e., $c\in \lbrace e_i, i=1,\ldots,n \rbrace$, as a defense strategy to achieve, apparently, the minimal cost. Then we obtain a collective set containing all $c$ that can enable a successful defense. 

\begin{theorem}\label{theorem:sensor sets in cis}
Given a stable cone-invariant system \eqref{eq:cis} over a cone $\mathcal{K}$, 
for an undetectable attack \eqref{eq:undetectable attack 1}  with $\forall b\in \Lambda_{\mathcal{K}}$,
the defense strategy with $ \forall c\in \Pi $
 is successful, where $\Pi=\left\lbrace e_j: e_j\in \mathrm{int}\,\mathcal{K}^{\ast}\right\rbrace $. Moreover, if $A\in \varpi'(\mathcal{K})$,   then $\Pi=\left\lbrace e_j:e_j\in \mathcal{K}^{\ast}\right\rbrace$.
\end{theorem}

To establish Theorem \ref{theorem:sensor sets in cis},  we need the following lemma.

\begin{lemma}\label{lemma:(sI-A)cone-invariance}
Given a Hurwitz stable matrix $A$ and a cone $\mathcal{K}$, for each $s_0\geq 0$, $(s_0I-A)^{-1}$ is a  cone-invariant matrix, i.e., $(s_0I-A)^{-1}\in \pi(\mathcal{K})$ if  and only if $A$ is a  cross-positive matrix over $\mathcal{K}$, i.e., $A\in \varpi(\mathcal{K})$.
\end{lemma}
\begin{proof}
Sufficiency.  
Consider a stable system $\dot{x}(t)=Ax(t)$ with $A\in \varpi(\mathcal{K})$.   
If $x(0)\in \mathcal{K}$, 
 through the cone-invariance,  we have
$x(t)\in \mathcal{K}, \forall\, t\geq 0 $.  
After taking the Laplace transform,  it follows that $X(s)=(sI-A)^{-1}x(0)$.   For each $s_0\geq 0$,   since $A$ is stable,   it  is
 larger than  the radius of convergence of the Laplace transform,  which
 implies that $X(s)$ does converge at point $s_0$.
 From 
$X(s_0)=\int_0^{\infty}x(t)e^{-s_0 t}\mathrm{d} t$,   together with the facts that
$x(t)\in \mathcal{K}$ and  $e^{-s_0 t}>0$,   it is easy to see that 
 $X(s_0)\in \mathcal{K},\,\forall\, s_0\geq 0$.
 Therefore,  for $\forall\, x(0)\in \mathcal{K}$,  we have $X(s_0)=(s_0I-A)^{-1}x(0)\in \mathcal{K}$,   which points out that  $(s_0I-A)^{-1}\in \pi(\mathcal{K}).$ 

Necessity. 
For $s_0=0$,   it can be easily verified that $A\in \varpi(\mathcal{K})$ if $-A^{-1}\in \pi(\mathcal{K})$.
For  $s_0> 0$,
since $(s_0I-A)^{-1}\in \pi(\mathcal{K})$,  it follows that $(I-\frac{1}{s_0}A)^{-1}\in \pi(\mathcal{K})$.  According to Definition \ref{definition:Cross-positive}, we define $f(k)=\left\langle y,(I-kA)^{-1}x \right\rangle $, where $x\in \mathcal{K}, y\in \mathcal{K}^{\ast}$ such that $\left\langle y,x \right\rangle=0$. Clearly,  it holds that $f(0)=0$.  In addition, since $(I-kA)^{-1}x\in \mathcal{K}$ and $y\in \mathcal{K}^{\ast}$,  it follows that $f(k)\geq 0,\,\forall\, k>0.$  Hence,   
denote $\dot{f}(k)$ as  the  first derivative of $f(k)$.   We have  $\dot{f}(0)\geq 0$.
On the other hand, after calculating $\dot{f}(k)$, we obtain 
$$\dot{f}(k)=\left\langle y,(I-kA)^{-1}A(I-kA)^{-1} x \right\rangle,$$
Thus,  we have
$\dot{f}(0)=\left\langle y,Ax \right\rangle\geq 0.$  The matrix $A$ is  cross-positive over $\mathcal{K}$, i.e., $A\in \varpi(\mathcal{K})$.   
\end{proof}

Now, we are ready to prove Theorem \ref{theorem:sensor sets in cis}.

\begin{proof}
For $s_0\geq 0$,  since $A$ is stable,  the matrix $s_0I-A$ has full rank and is invertible. Then,
\begin{equation}
\begin{aligned}
\!&\rank\!\left( \begin{bmatrix}
s_0I-A&b\\c^{\top}&0
\end{bmatrix}\right)\!=\!
\rank\!\left( \begin{bmatrix}
s_0I-A&0\\0&-c^{\top}(s_0I-A)^{-1}b 
\end{bmatrix}\right).
\end{aligned}
\nonumber
\end{equation} 
The  matrix above is  full  rank if and only if $c^{\top}(s_0I-A)^{-1}b\neq 0$ with $\forall s_0\geq 0$, $\forall\, b\in \Lambda$,  and $\forall\, c\in \Pi$. 
 
At first,  by Lemma \ref{lemma:(sI-A)cone-invariance},  we have $(s_0I-A)^{-1}\in \pi(\mathcal{K})$.  Since $b\in \Lambda_{\mathcal{K}} \subseteq \mathcal{K}$,   it follows that $(s_0I-A)^{-1}b \in \mathcal{K}$. In light of Lemma \ref{lemma:cone and its dual},   
we claim that $c^{\top}(s_0I-A)^{-1}b>0$ with $\forall\, c\in \left\lbrace e_j: e_j\in \mathrm{int}\,\mathcal{K}^{\ast}\right\rbrace$.

Consider the case $A\in \varpi'(\mathcal{K})$ and  $c\in \left\lbrace e_j: e_j\in \mathcal{K}^{\ast}\right\rbrace$.
From Definition \ref{definition:K-positive},  if $(s_0I-A)^{-1}$ is a $\mathcal{K}$-positive matrix, i.e., $(s_0I-A)^{-1} \in \pi^{+}(\mathcal{K})$,  we have $(s_0I-A)^{-1}b \in$ int$\mathcal{K}$, and thus $c^{\top}(s_0I-A)^{-1}b>0$ with the help of Lemma \ref{lemma:cone and its dual}.    Hence,  we need to verify that $(s_0I-A)^{-1} \in \pi^{+}(\mathcal{K}) $ with $A\in \varpi'(\mathcal{K})$ and $\forall s_0\geq 0$.

Assume that there exists $\bar{s}_0\geq 0$ such that $(\bar{s}_0I-A)^{-1}\notin \pi^{+}(\mathcal{K})$,  then there exist $l\in \mathcal{K}-\lbrace 0 \rbrace$ and $p\in \partial\mathcal{K}$ such that $p=(\bar{s}_0I-A)^{-1}l$,  i.e., $l=(\bar{s}_0I-A)p$. Choose $q\in \mathcal{K}^{\ast}$ such that $\left\langle q,p\right\rangle=0$. It follows that $\left\langle q,Ap\right\rangle\geq 0$. Therefore, we have
$$
0\leq \left\langle q,l\right\rangle=\left\langle q,(\bar{s}_0I-A)p\right\rangle=\bar{s}_0\left\langle q,p\right\rangle-\left\langle q,Ap\right\rangle\leq 0,
$$
indicating $\left\langle q,l\right\rangle=0$,  and then $l\in \partial\mathcal{K}$. Given $\bar{s}_1>\bar{s}_0$, denote 
$w=(\bar{s}_1I-A)p=(\bar{s}_1-\bar{s}_0)p+l$, we have $w\in \mathcal{K}$. Furthermore,  it follows that $w\in \partial\mathcal{K}$ since 
$$
\left\langle q,w\right\rangle=(\bar{s}_1-\bar{s}_0)\left\langle q,p\right\rangle+\left\langle q,l\right\rangle=0.
$$
In light of \cite{vandergraft1968spectral}, we can construct the face $\mathcal{F}_{w}$ of $\mathcal{K}$ generated by $w$ with 
\begin{equation}
\mathcal{F}_{\omega}=\left\lbrace x:  a x \preceq w,\, \exists a>0 \right\rbrace ,
\nonumber
\end{equation}
where $\preceq$ is defined on $\mathcal{K}$.
Since $w=(\bar{s}_1-\bar{s}_0)p+l\neq 0$,  $\forall x\in \mathcal{F}_{w}\subseteq \mathcal{K}$, we have  
\begin{equation}
\begin{aligned}
0&\leq \left\langle q,x\right\rangle=\left\langle q,\frac{1}{a}(a x- w)+\frac{1}{a}w\right\rangle
\\&=\frac{1}{a}\left\langle q,(a x- w)\right\rangle=-\frac{1}{a}\left\langle q,(w-a x)\right\rangle\leq 0.
\end{aligned}
\nonumber
\end{equation}
We obtain that $\left\langle q,x\right\rangle=0$ and $x\in$ $\partial\mathcal{K}$.
According to Definition \ref{definition:face}, the face $\mathcal{F}_{w}$ is nontrivial.
Next, we can obtain that
\begin{equation}
\begin{aligned}
0&\preceq (\bar{s}_1I-A)^{-1}z \preceq (\bar{s}_1I-A)^{-1}\frac{w}{a}=\frac{p}{a}
\\
&~~~~~~~~~~~~~~~~=\frac{w-l}{(\bar{s}_1-\bar{s}_0)a}\preceq \frac{w}{(\bar{s}_1-\bar{s}_0)a},
\end{aligned}
\nonumber
\end{equation}
which yields that $\mathcal{F}_{w}$ is invariant with respect to  $(\bar{s}_1I-A)^{-1}$, so is to $(\bar{s}_0I-A)^{-1}.$ According to Definition \ref{definition:face} and \ref{definition:irreducible cone-invariant matrix},   it can be shown that
$(\bar{s}_0I-A)^{-1}\notin \pi'(\mathcal{K})$.
In addition,  by Definition 
\ref{definition:irreducible matrix},  we can verify 
the existence of an eigenvector of $(\bar{s}_0I-A)^{-1}$ located on $ \partial\mathcal{K}$,  so is to  $A$ since $(\bar{s}_0I-A)^{-1}$ and $A$ share the same eigenvectors, which therefore leads to  contradiction.  Hence,   if $A$ is a  $\mathcal{K}$-irreducible matrix, $(s_0I-A)^{-1}$ is a $\mathcal{K}$-positive matrix for any $s_0\geq 0$.
\end{proof}

The following corollaries specialize the result in Theorem \ref{theorem:sensor sets in cis}  to  positive systems,  dynamic covariance systems,   and Lorentz cone-invariant systems.  Note that all related cones $\mathbb{R}_{+}^n$, $\mathcal{S}_{+}^n$ and $\mathcal{L}_{+}^n$ are self-dual and the corresponding  $\Lambda_{\mathbb{R}_{+}^n}, ~\Lambda_{\mathcal{S}_{+}^n},~ \Lambda_{\mathcal{L}_{+}^n}$ are specified in \eqref{eq:Lambda cone}.
\begin{corollary}\label{coro:stable positive systems}
Consider a stable positive system in \eqref{eq:cis} such that $A\!\in\! \varpi'(\mathbb{R}_{+}^n)$.  For an  undetectable attack \eqref{eq:undetectable attack 1} with $\forall b\in \Lambda_{\mathbb{R}_{+}^n}$,  the defense strategy with $\forall c\in \Pi=\Lambda_{\mathbb{R}_{+}^n}$ is successful.
\end{corollary}
For the positive cone $\mathbb{R}_{+}^n$, we have $\forall e_i \in \partial\mathbb{R}_{+}^n$. 
Since the matrix $A$ is $\mathbb{R}_{+}^n$-irreducible, i.e., $A\in \varpi'(\mathbb{R}_{+}^n)$, in light of Theorem \ref{theorem:sensor sets in cis},  we have
$\Pi=\Lambda_{\mathbb{R}_{+}^n}$. Therefore, for the stable positive system with an $\mathbb{R}_{+}^n-$irreducible matrix $A$, all the positions (states) are allowed to be attacked and  an arbitrary single measurement can be successful. 

\begin{corollary}\label{coro:stable dynamic covariance systems}
Consider a stable dynamic covariance system in \eqref{eq:cis} such that $A\in \varpi'(\mathcal{S}_{+}^n)$.  For an  undetectable attack \eqref{eq:undetectable attack 1}  with $\forall b\in \Lambda_{\mathcal{S}_{+}^n}$,  the defense strategy with $\forall c\in\Pi=\Lambda_{\mathcal{S}_{+}^n}$ is successful.
\end{corollary}
Not all Euclidean coordinates in $\mathbb{R}^{\frac{n(n+1)}{2}}$ belong to the PSD cone $\mathcal{S}_{+}^n$. We have $\Lambda_{\mathcal{S}_{+}^n}=\left\lbrace \mathrm{vec}(\mathrm{diag}(e_i)),i=1,\ldots,n \right\rbrace$  with $ \mathrm{vec}(\mathrm{diag}(e_i)) \in \partial\mathcal{S}_{+}^n$.  Hence,  since  $A\in \varpi'(\mathcal{S}_{+}^n)$, it follows that $\Pi=\Lambda_{\mathcal{S}_{+}^n}$.

\begin{corollary}\label{coro:stable Lorentz cone-invariant systems}
Consider a stable Lorentz cone-invariant system in \eqref{eq:cis} such that $A\in \varpi(\mathcal{L}_{+}^n)$.  For an undetectable attack \eqref{eq:undetectable attack 1} with $\forall b\in \Lambda_{\mathcal{L}_{+}^n}$, the defense strategy with $\forall c\in\Pi=\Lambda_{\mathcal{L}_{+}^n}$ is successful.
\end{corollary}
For the Lorentz cone $\mathcal{L}_{+}^n$, there exists only $e_1\in \mathcal{L}_{+}^n$, moreover $e_1\in \mathrm{int}\,\mathcal{L}_{+}^n$. Then we have  $\Pi=\Lambda_{\mathcal{L}_{+}^n}=\lbrace e_1\rbrace$.   For Lorentz cone-invariant systems, only the position labeled by $e_1$ can be attacked and the same position should also be measured to achieve  the protection.

It is worth emphasizing that existing works mainly focus on the protection of static systems, like power systems, in which undetectable attacks should hide in the kernel space of the studied matrix, projecting attack signals into the measurement residual.  This work extends this protection scheme into dynamic systems and prevents undetectable attacks corresponding to non-minimum phase zeros.	


\subsection{$\!\!$Marginally Stable Cone-Invariant System}
$\!\!$An LTI system is marginally stable if and only if the real part of every eigenvalue  of the matrix $A$ is non-positive and at least one simple eigenvalue is located on the imaginary axis.  
In this section, we extend the preceding  defensive strategies to marginally stable systems.

\begin{theorem}\label{theorem:sensor sets in marginally  cis}
Given a marginally stable cone-invariant system  \eqref{eq:cis} over a cone $\mathcal{K}$,   if $A\in \varpi'(\mathcal{K})$, for an undetectable attack \eqref{eq:undetectable attack 1} with  $\forall b\in  \Lambda_{\mathcal{K}}$,  the defense strategy with $\forall c\in\Pi$ is successful, where
$\Pi=\left\lbrace e_j:e_j\in \mathcal{K}^{\ast}\right\rbrace. $
\end{theorem}
\begin{proof}
For $s_0>0$,  the proof is technically similar to that for Theorem \ref{theorem:sensor sets in cis},  thus, omitted.
For $s_0=0$, in light of Lemma \ref{lemma:eigenvalue and eigenvector of cross-positive},  there exists a zero eigenvalue of 
$A\in \varpi'(\mathcal{K})$  and $A$ is not invertible.
Consider the matrix $[s_0-A \,\,b\,;\,c^{\top}\,0]=[-A \,\,e_i\,;\,e_j^{\top}\,0]$. 
If $e_i\notin \mathrm{Im}(A)\cap\Lambda_{\mathcal{K}}$,
we can easily verify that either the matrix $[-A \,\,e_i\,;\,e_j^{\top}\,0]$ has full rank or it is not full rank but $d_0=0$,  satisfying the conditions (i) and (ii) in Definition \ref{definition:successful defense strategy}, respectively. 
Assume that there exists  $e_i\in \mathrm{Im}(A)\cap\Lambda_{\mathcal{K}}$,  which also indicates 
a vector $\xi$ such that $e_i=A\xi$. 
From Lemma \ref{lemma:eigenvalue and eigenvector of cross-positive},  there  exists a vector
 $y\in \mathrm{int}\,\mathcal{K}^{\ast}$ satisfying $y^{\top}A=0$. 
 Therefore, it follows that 
$$\left\langle y,\,e_i\right\rangle=\left\langle y,\,A\xi\right\rangle=\left\langle yA^{\top}\!\!,\,\xi\right\rangle=0. $$
On the other hand,  
since $y\in \mathrm{int}\,\mathcal{K}^{\ast}$ and $e_i\in \Lambda_{\mathcal{K}}\in\mathcal{K}$,  by Lemma \ref{lemma:cone and its dual},  we have $\left\langle y,\,e_i\right\rangle>0$, which leads to contradiction.  
\end{proof}

\section{Multi-Agent Systems}  \label{sec:5}
In this section, we consider sensor placement strategies for multi-agent systems against undetectable attacks. 
We show that first-order multi-agent systems can be formulated as a special case of marginally stable positive systems, whereas second-order multi-agent systems are not. 
Nevertheless,  the features of  positive cone $\mathbb{R}_{+}^n$ still shed  useful lights into the defense strategy of 
second-order systems.
\subsection{Algebraic Graph Theory}
Let $\mathcal{G} = (\mathcal{V}, \mathcal{E}, \mathcal{A})$ denote a weighted digraph with the set of agents (vertices or nodes) $\mathcal{V}=\lbrace v_1,\ldots,v_n\rbrace$, the set of edges (links) $\mathcal{E}=\mathcal{V}\times \mathcal{V}$, and the weighted adjacency matrix $A=(a_{ij})$ with nonnegative elements $a_{ij}$. An edge of $\mathcal{G}$ is denoted by $(v_i,v_j)\in \mathcal{E}$ if there exists a directed link from agent $v_i$ to agent $v_j$ with weight $a_{ij}>0$, i.e., agent $v_j$ can receive information from agent $v_i$.  Denote $\mathcal{N}_i = \left\lbrace j \in \mathcal{V} : (v_j,v_i) \in \mathcal{E}\right\rbrace $ as the neighborhood set of node $i$.
A path from node $v_1$
to node $v_k$ is a sequence of nodes $v_1,\ldots,v_k$, such that for
each $i$, $1 \leq i \leq k-1$, $(v_i,v_{i+1})$ is an edge.
The in-degree of agent $v_i$ is defined as $\mathrm{deg}_{i}^{\mathrm{in}}=\sum_{j=1}^{n}a_{ij}.$ Define 
for the graph $\mathcal{G}$  the Laplacian matrix $L =\mathcal{D}-A$, where $\mathcal{D}\triangleq \mathrm{diag}(\mathrm{deg}_{1}^{\mathrm{in}},\ldots,\mathrm{deg}_{n}^{\mathrm{in}})$ is the in-degree matrix.   Its row sums equal to zero.
Note that any negative Laplacian matrix $-L$ is a Metzler matrix,  which is cross-positive over the positive cone $\mathbb{R}_{+}^n$.
And it follows that $L$ can be decomposed as
 $L=\alpha I-\bar{L}$, where $\alpha=\rho(\bar{L})$ is the spectral radius of $\bar{L}$, and $\bar{L}\in \pi(\mathbb{R}_+^n)$.

Connectivity is one of the basic concepts in graph theory.
A digraph is strongly connected if it contains a directed path for every pair of vertices. Meanwhile, the Laplacian matrix $L$ is irreducible if and only if its corresponding digraph is strongly connected. 
Different from the strong connectivity, a
 weaker connectivity  requires a digraph to have a directed spanning tree, by which we mean that a certain 
agent $v_i$ exists such
that for any other agent, there is a directed path from 
 $v_i$ to that agent.
 The agent $v_i$ is called a root of the spanning tree.
A digraph can be broken down into strongly connected components
and every component is a  strongly connected subgraph.
Consider a digraph with $M$ strongly connected
components. The $i$-th strongly connected
component  is denoted as $SCC_i,\, i = 1,\ldots, M$ with $n_i$ nodes.  
The Laplacian matrix in this case can be represented by a block triangular matrix,  with a proper arrangement of the nodes:
$$L= \begin{bmatrix}
  L_1&\ast&\ast&\ast\\
  0&L_2&\ast&\ast\\
  \vdots&\vdots&\ddots&\vdots\\
  0&0&0&L_{M}
  \end{bmatrix}.
$$
In this form,  $L_i\in \mathbb{R}^{n_i\times n_i}$.
Clearly,  if the digraph contains  spanning trees,  all feasible roots are located in the first component  $SCC_1$ associated with the submatrix $L_1$  above.
For example, Fig. \ref{fig:digraph two components} shows a digraph of 7 agents,   divided into two strongly connected components. The corresponding Laplacian matrix is given by
\begin{equation}\label{eq:Laplacian}
L= \begin{bmatrix}
    \begin{array}{ccc|ccc}
  1.5&-1.5&0&0&0&0 \\
  0&2.8&-0.3&-2.5&0&0\\
  -2&0&3.5&-1.5&0&0\\
  \hline
 0&0&0&0.1&-0.1&0\\
 0&0&0&0&1&-1\\
 0&0&0&-2.7&0&2.7
    \end{array}
  \end{bmatrix}.
\end{equation}  
\vspace{0.5cm}

\begin{figure}[H]
\centering
\includegraphics[width=7.5cm, height=5.3cm]{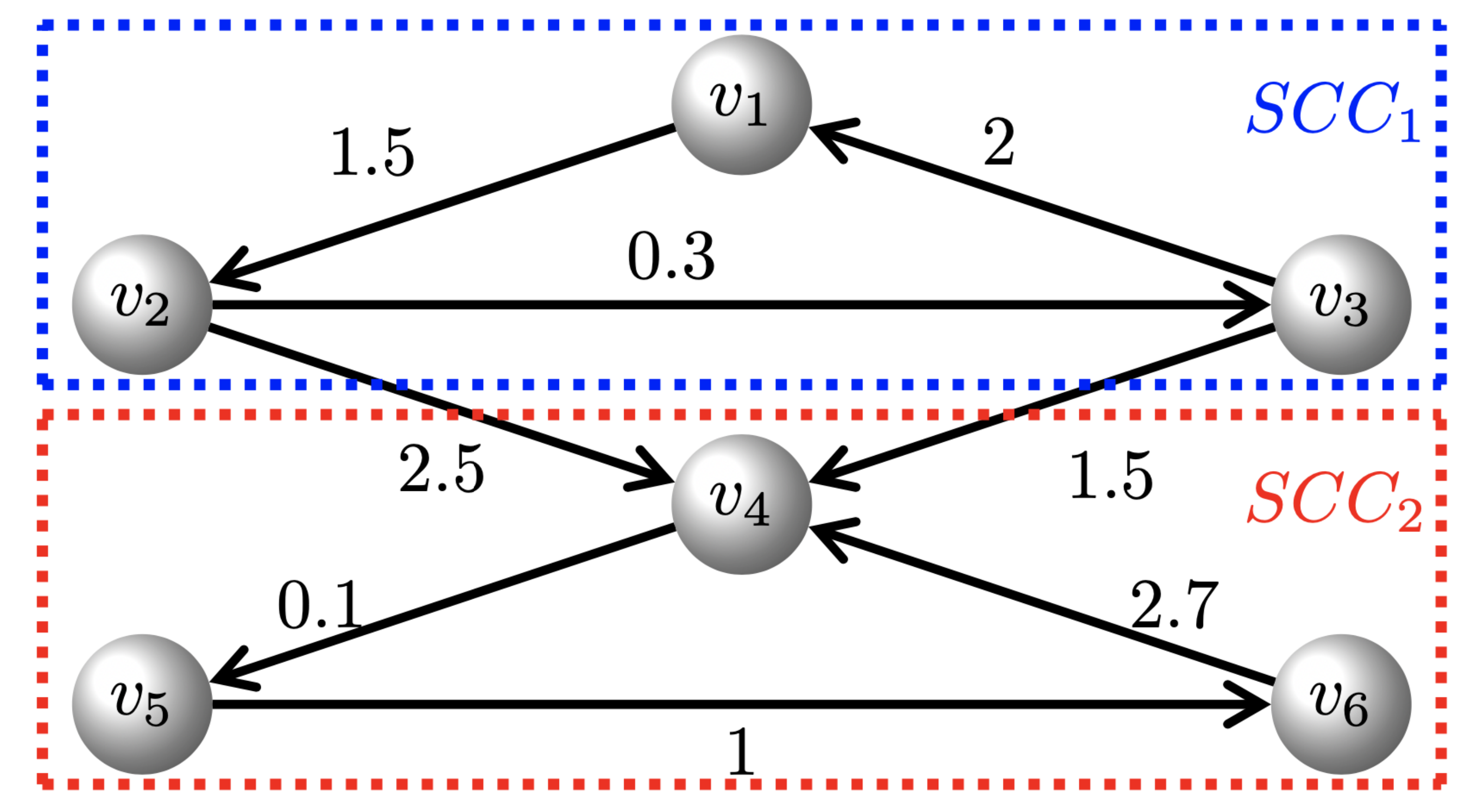}
	\caption{A digraph contains two strongly connected components}
		\label{fig:digraph two components}
		\end{figure}
		
\subsection{Single-Integrator Agent Dynamics}
Consider the single-integrator 

\noindent agents with state-space equations given by
\begin{equation}\label{eq:single integrator}
\dot{x}_i(t)=u_i(t),
\end{equation}
together with 
 the  distributed consensus protocol 
\begin{equation}\label{eq:single integrator distributed consensus protocol}
u_i(t)=-\sum\limits_{j\in \mathcal{N}_i} (x_i(t)-x_j(t)),
\end{equation}
where $x_i \in \mathbb{R}$ is the state and $u_i\in \mathbb{R}$ is the control law. 
Then the closed-loop system combining \eqref{eq:single integrator} and \eqref{eq:single integrator distributed consensus protocol} under an undetectable attack can be written as:
\begin{equation}\label{eq: first-order multi-agent system}
\begin{aligned}
\dot{x}(t)&=-Lx(t)+bd(t),\\
y(t)&=c^{\top}x(t),
\end{aligned}
\end{equation}
where $x=[x_1,\ldots, x_n]^{\top}$, and $L$ is the Laplacian matrix of the corresponding digraph. Note that this first-order multi-agent system is a  marginally stable cone-invariant system over the positive cone $\mathbb{R}_{+}^n$.  When the digraph is strongly connected, 
in view of  Remark \ref{remark:positive irreducible and irreducible matrix},
$L$ is irreducible
and thus is $\mathbb{R}_+^n$-irreducible.  According to Theorem \ref{theorem:sensor sets in marginally  cis}, we have the following result immediately.
\begin{corollary}\label{coro:first-order strongly connected multi-agent}
Consider the first-order multi-agent system  \eqref{eq: first-order multi-agent system} defined on a strongly connected digraph.  For an  undetectable attack \eqref{eq:undetectable attack 1} with $\forall b\in\Lambda_{\mathbb{R}_+^n}$,  the defense strategy with $\forall c\in\Pi$ is successful, where $\Pi=\Lambda_{\mathbb{R}_+^n}$.
\end{corollary}

Moreover, going beyond Theorem \ref{theorem:sensor sets in marginally  cis}, we can generalize Corollary \ref{coro:first-order strongly connected multi-agent} by weakening the strong connectivity requirement   of the digraph in the following theorem.

\begin{theorem}\label{theorem:first-order multi-agent 2}
Consider the first-order multi-agent system \eqref{eq: first-order multi-agent system} defined on a digraph. 

\noindent (1) For any $SCC_i$, define the set 
\begin{equation}\label{eq:SCC_i 1}
\Lambda_i=\lbrace e_j: \sum_{l=1}^{i-1} n_l < j\leq \sum_{l=1}^{i} n_l \rbrace,
\end{equation}
in which $\Lambda_1=\lbrace e_j:0<j\leq n_1\rbrace$.
For an  undetectable attack \eqref{eq:undetectable attack 1} with $\forall b\in\Lambda_i$, the defense strategy with $\forall c\in\Pi_i$ is almost successful, where $\Pi_i=\Lambda_i$. 

\noindent (2) If the digraph contains spanning trees, 
for an  undetectable attack \eqref{eq:undetectable attack 1} with $\forall b\in\Lambda_{\mathbb{R}_+^n}$,  the defense strategy with $\forall c\in\Pi$ is almost successful, where $\Pi=\Lambda_1$.
\end{theorem}
\begin{proof}
Note that,  for a general digraph,   the corresponding Laplacian matrix $L$ may be reducible.   For $s_0>0$,
in light of Lemma \ref{lemma:(sI-A)cone-invariance}, we have $(s_0I+L)^{-1}\in \pi(\mathbb{R}_+^n)$.  
We decompose $L$ as
$L=\alpha I-\bar{L}$ with $\alpha=\rho\left(\bar{L} \right) $ and $\bar{L}\in \pi(\mathbb{R}_+^n)$.
 Then it follows that
\begin{equation}
(s_0I+L)^{-1}=((s_0+\alpha)I-\bar{L})^{-1}=\frac{1}{s_0+\alpha}(I-\frac{1}{s_0+\alpha}\bar{L})^{-1}.
\nonumber
\end{equation} 
From $\rho(\frac{1}{s_0+\alpha}\bar{L})<1$, we can refer to Taylor series and get
\begin{equation}
\frac{1}{s_0+\alpha}(I-\frac{1}{s_0+\alpha}\bar{L})^{-1}=\frac{1}{s_0+\alpha}\sum\limits_{k=0}^{\infty} \left( \frac{1}{s_0+\alpha}\bar{L}\right)^k.
\nonumber
\end{equation} 
For $\bar{L}\in \pi(\mathbb{R}_+^n)$,  if $\bar{L}_{x,\,y}>0$, the $(x,y)$-th entry of $\bar{L}$,  it implies that there exists an edge $(v_x,v_y)$ in the digraph.  Thus, 
one can claim that $\bar{L}^{k}_{x,\,y}>0$ if and only if there exists at least one path from node $v_x$ to node $v_y$ within $k$ steps.

In the case (1),  in view of the strong connectivity of $SCC_i$,  then for any two nodes $v_x,\,v_y$ in $SCC_i$, 
there exists at least one directed path  with not more than $n_i$ steps,  
yielding that
$\sum_{k=0}^{\infty} \bar{L}^k_{x,\,y}>0$, so is to $\frac{1}{s_0+\alpha}(I-\frac{1}{s_0+\alpha}\bar{L})^{-1}_{x,\,y}$.  
Consequently,  we have 
$c^{\top}(s_0I+L)^{-1}b>0$ with $\forall s_0> 0$, $\forall\, b\in \Lambda_i$,  and $\forall\, c\in \Pi_i$. 

In the case (2),  the sets $\Lambda_{\mathbb{R}_+^n}$ and  $\Pi=\Lambda_1$ can be determined in a similar  manner by noting that
$SCC_1$ contains all roots of spanning trees in the digraph.
\end{proof}


\subsection{Double-Integrator Agent Dynamics}
Consider the double-integrator agents with state-space equations described by
\begin{equation}\label{eq:second-order multi-agent-plant}
\begin{aligned}
\dot{\xi}_i(t)&=\zeta_{i}(t)\\
\dot{\zeta}_i(t)&=u_i(t),
\end{aligned}
\end{equation}
where $\xi_i\in\mathbb{R}$ and $\zeta_i\in\mathbb{R}$ are the corresponding states and $u_i\in \mathbb{R}$ denotes the distributed control law. 
Next,  we extend our preceding sensor placement strategies into 
second-order multi-agent systems which are
 not invariant with the positive cone $\mathbb{R}_+^n$. 
Before showing main results of this section,  we provide the following lemma at first.
\begin{lemma}\label{lemma:Laplacian e_i}
For a Laplacian matrix $L$ defined on a strongly connected digraph, we have  
$e_i \notin \mathrm{Im}(L)$ and $e_i \notin \mathrm{Im}(L^{\top})$,  $\forall e_i, i\in \lbrace 1,\ldots,n \rbrace$.
\end{lemma}
\begin{proof}
We  prove the lemma by contradiction.   Assume that there exists $e_i$ such that $e_i\in \mathrm{Im}(A)$.  Then
we can obtain  $x \neq 0$ satisfying $e_i = Lx$. In light of Lemma \ref{lemma:eigenvalue and eigenvector of cross-positive}, there exists a strictly positive vector $v$ such that $v^{\top}L=0$. 
It follows that $v^{\top}e_i=v^{\top}Lx=0$,
showing that the $i$th entry of $v$ is zero.  This fact contradicts with the strict  positivity of $v$.  Also,  the result $e_i \notin \mathrm{Im}(L^{\top})$ can be proved in a similar manner. 
\end{proof}

Consider two general types of distributed control laws,  respectively.
The first class of distributed control law is designed as
\begin{equation}
\label{eq:control law 1}
u_i(t)=-k_i\zeta_{i}(t)+\sum_{j\in N_i}\upsilon _{i,j}(\xi_j(t)-\xi_i(t)),
\end{equation}
where $k_i,\,\upsilon_{i,j}>0,\; \forall i,j=1,\ldots,N$ \cite{shames2011distributed,chen2017protecting}. 
One available physical interpretation of $\xi_i(t)$ and $\zeta_{i}(t)$ are the position and the velocity of the agent $i$ in a mechanical system, and $k_i$ can be interpreted as the radio of agent's damper to mass. The agents under the control law \eqref{eq:control law 1} move toward the position of their neighbors while damping their current velocity \cite{van2013port}.
The closed-loop system combining (\ref{eq:second-order multi-agent-plant}) and (\ref{eq:control law 1})  under an undetectable attack can be written as:
\begin{equation}\label{eq:second-order multi-agent-closedloop 1}
\begin{aligned}
\begin{bmatrix}
\dot{\xi}(t)\\\dot{\zeta}(t)
\end{bmatrix}&=\begin{bmatrix}
0&I\\
-L&-K
\end{bmatrix}\begin{bmatrix}
\xi(t)\\ \zeta(t)
\end{bmatrix}+bd(t),\\
y(t)&=c^{\top}\begin{bmatrix}
\xi(t)\\ \zeta(t)
\end{bmatrix},
\end{aligned}
\end{equation}
where $\xi(t)=[\xi_1(t),\ldots, \xi_n(t)]^{\top}$, $\zeta(t)=[\zeta_1(t),\ldots, \zeta_n(t)]^{\top}$,  $L$ is the Laplacian matrix,  and $K=\mathrm{diag}\left\lbrace  k_1,\ldots ,k_n\right\rbrace $. Denote the canonical  basis in $\mathbb{R}^{2n}$ as $\lbrace f_i, i=1,\ldots,2n\rbrace$, i.e., $f_j\!=\![e_j^{\top} \,\textbf{0}^{\top}]^{\top}$ and $f_{j+n}\!=\![\textbf{0}^{\top} \,e_j^{\top}]^{\top}\!\!, j=1,\ldots,n$,  where $\textbf{0}$ denotes an $n$-dimensional zero vector.
\begin{theorem}\label{theorem:secondt-order strongly connected multi-agent 1}
Consider the second-order multi-agent system \eqref{eq:second-order multi-agent-closedloop 1}  defined on a strongly connected digraph.  For an  undetectable attack  \eqref{eq:undetectable attack 1} 
with $\forall b\in\Lambda_{\mathbb{R}_{+}^{2n}}$,  the defense strategy with $\forall c\in\Pi$ is successful, where $\Pi=\Lambda_{\mathbb{R}_{+}^{2n}}=\lbrace f_i, i=1,\ldots,2n\rbrace$.
\end{theorem}
\begin{proof}
We start by the case $s_0>0$.
Consider all the four possible combinations of $b$ and $c$
 with  $b=\begin{bmatrix}
e_i^{\top}&\textbf{0}^{\top}
\end{bmatrix}^{\top}$ or $b=\begin{bmatrix}
\textbf{0}^{\top}&e_i^{\top}
\end{bmatrix}^{\top}$,  and   $c=\begin{bmatrix}
e_j^{\top}&\textbf{0}^{\top}
\end{bmatrix}^{\top}$ or $c=\begin{bmatrix}
\textbf{0}^{\top}&e_j^{\top}
\end{bmatrix}^{\top}\!\!,\forall e_i,e_j :$

\begin{enumerate}[(i)]
\item If $b=\begin{bmatrix}
e_i^{\top}&\textbf{0}^{\top}
\end{bmatrix}^{\top}$ and  $c^{\top}=\begin{bmatrix}
e_j^{\top}&\textbf{0}^{\top}
\end{bmatrix}$, we have
\begin{equation}
\begin{aligned}
&\rank\left( \begin{bmatrix}
s_0I-A&b\\c^{\top}&0
\end{bmatrix}\right)
=\rank\left( \begin{bmatrix}
s_0I&-I&e_i\\L&s_0I+K&\textbf{0}\\e_j^\top &\textbf{0}^{\top}&0
\end{bmatrix}\right)  
\\
&=\rank\left( \begin{bmatrix}
\textbf{0}_{n\times n}&-I&e_i\\I&\textbf{0}_{n\times n}&(s_0I+K)e_i\\ \textbf{0}^{\top} &\textbf{0}^{\top}& -e_j^\top (s_0^2I\!+\!s_0K\!+\!L)^{-1}(s_0I\!+\!K)e_i
\end{bmatrix}\right),
\end{aligned}
\nonumber
\end{equation}

\noindent where $\textbf{0}_{n\times n}\in \mathbb{R}^{n\times n}$ is the zero matrix.
The matrix $s_0^2I+s_0K+L$ is invertible since $s_0>0$ and $K=\mathrm{diag}\left( k_1,\ldots ,k_n\right)$ is a strictly positive diagonal matrix.
Therefore, it holds that
\begin{equation}\label{e:rank-second-order-system}
\rank \left( \begin{bmatrix}
s_0I-A&b\\c^{\top}&0
\end{bmatrix}\right)=2n+1
\end{equation} 
if and only if
\begin{center}
$e_j^\top (s_0^2I+s_0K+L)^{-1}(s_0I+K)e_i \neq 0,\, \forall e_i,e_j$.
\end{center} 

\item If $b=\begin{bmatrix}
\textbf{0}^{\top}&e_i^{\top}
\end{bmatrix}^{\top}$ and  $c^{\top}=\begin{bmatrix}
e_j^{\top}&\textbf{0}^{\top}
\end{bmatrix}$, we have that \eqref{e:rank-second-order-system} holds 
if and only if
\begin{center}
$e_j^\top (s_0^2I+s_0K+L)^{-1}e_i \neq 0,\, \forall e_i,e_j$.
\end{center}

\item If $b=\begin{bmatrix}
\textbf{0}^{\top}&e_i^{\top}
\end{bmatrix}^{\top}$ and  $c^{\top}=\begin{bmatrix}
\textbf{0}^{\top}&e_j^{\top}
\end{bmatrix}$, we have that \eqref{e:rank-second-order-system} holds 
if and only if
\begin{center}
$e_j^\top s_0(s_0^2I+s_0K+L)^{-1}e_i \neq 0,\, \forall e_i,e_j$,
\end{center}

\item If $b=\begin{bmatrix}
e_i^{\top}&\textbf{0}^{\top}
\end{bmatrix}^{\top}$ and  $c^{\top}=\begin{bmatrix}
\textbf{0}^{\top}&e_j^{\top}
\end{bmatrix}$, we have that \eqref{e:rank-second-order-system} holds 
if and only if
\begin{center}
$e_j^\top\left(I-s_0(s_0^2I+s_0K+L)^{-1}(s_0I+K)\right) e_i \neq 0,\, \forall e_i,e_j$.
\end{center}
\end{enumerate}
The cases (i), (ii), and (iii) require us to prove that 
all entries of the matrix $(s_0^2I+s_0K+L)^{-1}$ are non-zero, whilst we need to clarify that
the matrix $I-s_0(s_0^2I+s_0K+L)^{-1}(s_0I+K)$ holds no zero entry
for the case (iv).

For the cases (i), (ii), and (iii), one can easily obtain that $-K-(1/s_0)L$ is a stable matrix and we omit the proof. We then prove that all entries of $(s_0^2I+s_0K+L)^{-1}$ are non-zero via proving that $-K-(1/s_0)L$ is a Metzler matrix, i.e., $-K-(1/s_0)L\in \varpi(\mathbb{R}_+^n)$.  
According to Definition \ref{definition:Cross-positive},  $\forall x,\,y \in \mathbb{R}_+^n$ such that
$\left\langle y,x \right\rangle=0$,
it follows that 
 $$ \left\langle y,(-K-(1/s_0)L)x \right\rangle=(1/s_0)\left\langle y,-Lx \right\rangle-\left\langle y,Kx \right\rangle. $$ 
 Since $-L\in \varpi(\mathbb{R}_+^n)$, we have $\left\langle y,-Lx \right\rangle\geq 0$. 
 On the other hand,   we can easily obtain that $\left\langle y,Kx \right\rangle=0$ on account of the fact that $\left\langle y,x \right\rangle=0$ and the strict positiveness of the diagonal matrix $K$.  As such,  by noting that
 $\left\langle y,(-K-(1/s_0)L)x \right\rangle \geq 0,$
  we claim that 
 $-K-(1/s_0)L\in \varpi(\mathbb{R}_+^n)$.  Additionally,
 since 
 $L$ is  irreducible,   we have  $-K-(1/s_0)L\in \varpi'(\mathbb{R}_+^n)$.  
Hence,  
 all the entries of $(s_0^2I+s_0K+L)^{-1}$ are strictly positive by Theorem \ref{theorem:sensor sets in cis}.



For the case (iv),  we  consider the matrix $I-s_0(s_0^2I+s_0K+L)^{-1}(s_0I+K)$. 
Since all the entries of $(s_0^2I+s_0K+L)^{-1}$ are strictly positive, 
the matrix $s_0(s_0^2I+s_0K+L)^{-1}(s_0I+K)$ is also a strictly positive matrix.  
One can easily see that
\begin{equation}
s_0(s_0^2I+s_0K+L)^{-1}(s_0I+K)=
\left(I+(s_0^2I+s_0K)^{-1}L \right)^{-1},
\nonumber
\end{equation}
in which  the matrix
$(s_0^2I+s_0K)^{-1}L$, denoted as $\tilde{L}$,  is also a Laplacian matrix defined on a strongly connected digraph. 
Together with the fact that $(I+\tilde{L})^{-1}$ is a row stochastic matrix, which is implied by that $(I+\tilde{L})$ is a row stochastic matrix,
it follows that  the $(i,j)$-th entry $(I+\tilde{L})^{-1}_{i,j}\in\left(0,1 \right), \forall i,j,$  which further leads to
\begin{equation}
\begin{aligned}
\left(I-s_0(s_0^2I+s_0K+L)^{-1}(s_0I+K)\right)_{i,i}&>0,\\
\left(I-s_0(s_0^2I+s_0K+L)^{-1}(s_0I+K)\right)_{i,j}&<0, \, \forall i\neq j.
\end{aligned}
\nonumber
\end{equation}

Finally,  we study the case $s_0=0$.
When $b=\begin{bmatrix}
e_i^{\top}&\textbf{0}^{\top}
\end{bmatrix}^{\top}$ or $b=\begin{bmatrix}
\textbf{0}^{\top}&e_i^{\top}
\end{bmatrix}^{\top}\!\!,$ and $c^{\top}=\begin{bmatrix}
e_j^{\top}&\textbf{0}^{\top}
\end{bmatrix}$, 
the matrices   
\begin{equation}
\begin{bmatrix}
\textbf{0}_{n\times n}&-I&e_i\\
L&K&\textbf{0}\\
e_j^{\top}&\textbf{0}^{\top}&0
\end{bmatrix}
\quad
\mathrm{or}
\quad
\begin{bmatrix}
\textbf{0}_{n\times n}&-I&\textbf{0}\\
L&K&e_i\\
e_j^{\top}&\textbf{0}^{\top}&0
\end{bmatrix},
\end{equation}
are both full rank  in view of the fact, from Lemma \ref{lemma:Laplacian e_i},  that $e_i \notin \mathrm{Im}(L),\,e_j \notin \mathrm{Im}(L^{\top}),\forall e_i,e_j$.
When $b=\begin{bmatrix}
e_i^{\top}&\textbf{0}^{\top}
\end{bmatrix}^{\top}$ or $b=\begin{bmatrix}
\textbf{0}^{\top}&e_i^{\top}
\end{bmatrix}^{\top}$,  and $c^{\top}=\begin{bmatrix}
\textbf{0}^{\top}&e_j^{\top}
\end{bmatrix}$,  clearly,   the  matrices 
\begin{equation}\label{e:RSM-a=0}
\begin{bmatrix}
\textbf{0}_{n\times n}&-I&e_i\\
L&K&\textbf{0}\\
\textbf{0}^{\top}&e_j^{\top}&0
\end{bmatrix}
\quad
\mathrm{or}
\quad
\begin{bmatrix}
\textbf{0}_{n\times n}&-I&\textbf{0}\\
L&K&e_i\\
\textbf{0}^{\top}&e_j^{\top}&0
\end{bmatrix},
\end{equation} 
are  not full  rank.   
However,  
since $e_i \notin \mathrm{Im}(L),\forall e_i$,    recall the equation \eqref{eq:rosenbrock}, 
we can arrive at $d_0=0$ directly.
The proof is now completed  by invoking the condition (ii) in Definition \ref{definition:successful defense strategy}.
\end{proof}

Consider the second type of distributed control law given by
\begin{equation}
\label{eq:control law 2}
u_i(t)=r\sum_{j\in N_i}\upsilon _{i,j}(\zeta_j(t)-\zeta_i(t))+\sum_{j\in N_i}\upsilon _{i,j}(\xi_j(t)-\xi_i(t)),
\end{equation}
where $r>0$ and $\upsilon_{i,j}>0,\; \forall i,j=1,\ldots,N$ \cite{yu2010some,hong2008distributed}. 
One motivation for such distributed control law comes from power systems, where $\xi_i(t)$ and $\zeta_i(t)$ are interpreted as the phase and the frequency of the agent $i$. Under this control law,   the power system reaches consensus  while not only the phase differences but also the frequency difference are penalized \cite{kundur1994power}. 
The compact form of the closed-loop system combining (\ref{eq:second-order multi-agent-plant}) and (\ref{eq:control law 2}) under an undetectable attack is given by
\begin{equation}\label{eq:second-order multi-agent-closedloop 2}
\begin{aligned}
\begin{bmatrix}
\dot{\xi}(t)\\\dot{\zeta}(t)
\end{bmatrix}&=\begin{bmatrix}
0&I\\
-L&-rL
\end{bmatrix}\begin{bmatrix}
\xi(t)\\ \zeta(t)
\end{bmatrix}+bd(t),\\
y(t)&=c^{\top}\begin{bmatrix}
\xi(t)\\ \zeta(t)
\end{bmatrix}.
\end{aligned}
\end{equation}
\begin{theorem}\label{theorem:secondt-order strongly connected multi-agent 2}
Consider a second-order multi-agent system \eqref{eq:second-order multi-agent-closedloop 2} defined on a strongly connected digraph. 

\noindent (1) For an undetectable attack  \eqref{eq:undetectable attack 1} 
 with $\forall b\in \Lambda_1=\lbrace f_i, i=1,\ldots,n \rbrace$,  the defense strategy with $\forall c\in \Pi$ is almost successful, where $\Pi=\Lambda_{\mathbb{R}_{+}^{2n}}$.

\noindent (2) For an  undetectable attack  \eqref{eq:undetectable attack 1} 
with $\forall b\in\Lambda_2=\lbrace f_i, i\!=\!n+1,\ldots,2n \rbrace$, 
the defense strategy with $\forall c\in\Pi$ is successful, also $\Pi=\Lambda_{\mathbb{R}_{+}^{2n}}$.

\end{theorem}
\begin{proof}
Consider the case $s_0>0$. Likewise,  from the four possible combinations of $b$ and $c$,   we have:
\begin{enumerate}[(i)]
\item If $b=\begin{bmatrix}
e_i^{\top}&\textbf{0}^{\top}
\end{bmatrix}^{\top}$ and  $c^{\top}=\begin{bmatrix}
e_j^{\top}&\textbf{0}^{\top}
\end{bmatrix}$, we need to verify that
\begin{center}
$e_j^\top (s_0^2I+(s_0r+1)L)^{-1}(s_0I+rL)e_i \neq 0,\, \forall e_i,e_j$.
\end{center} 

\item If $b=\begin{bmatrix}
\textbf{0}^{\top}&e_i^{\top}
\end{bmatrix}^{\top}$ and  $c^{\top}=\begin{bmatrix}
e_j^{\top}&\textbf{0}^{\top}
\end{bmatrix}$, we need to verify that
\begin{center}
$e_j^\top (s_0^2I+(s_0r+1)L)^{-1}e_i \neq 0,\, \forall e_i,e_j$.
\end{center}

\item If $b=\begin{bmatrix}
\textbf{0}^{\top}&e_i^{\top}
\end{bmatrix}^{\top}$ and  $c^{\top}=\begin{bmatrix}
\textbf{0}^{\top}&e_j^{\top}
\end{bmatrix}$, we need to verify that
\begin{center}
$e_j^\top s_0(s_0^2I+(s_0r+1)L)^{-1}e_i \neq 0,\, \forall e_i,e_j$,
\end{center}

\item If $b=\begin{bmatrix}
e_i^{\top}&\textbf{0}^{\top}
\end{bmatrix}^{\top}$ and  $c^{\top}=\begin{bmatrix}
\textbf{0}^{\top}&e_j^{\top}
\end{bmatrix}$, we need to verify that
\begin{center}
\hspace{-0.85cm}
$e_j^\top\left(I-s_0(s_0^2I+(s_0r+1)L)^{-1}(s_0I+rL)\right) e_i \neq 0,\, \forall e_i,e_j$.
\end{center}
\end{enumerate}
The cases (i) and (iv) represent $\Lambda_1$ and the cases (ii) and (iii) point to $\Lambda_2$.  

For the cases (ii) and (iii),  it follows as in the proofs of Theorem \ref{theorem:sensor sets in cis} and \ref{theorem:secondt-order strongly connected multi-agent 1} that
all the entries of the matrix $(s_0^2I+(s_0r+1)L)^{-1}$ are non-zero.

For the case (i),  we study the matrix
the matrix $(s_0^2I+(s_0r+1)L)^{-1}(s_0I+rL)=\left(s_0I+(s_0I+rL)^{-1}L\right) ^{-1}$.
We first  show  that $(s_0I+rL)^{-1}L$ is Laplacian.  Clearly, 
 its row sums equal to zero.  
Let $L=\alpha I-\bar{L}$ with $\alpha=\rho(\bar{L})$ and $\bar{L}\in \pi(\mathbb{R}_{+}^n)$. It follows that
\begin{equation}
\begin{aligned}
(s_0I+rL)^{-1}L=\frac{1}{s_0+r\alpha}\left(I-\frac{r}{s_0+r\alpha}\bar{L} \right) ^{-1}(\alpha I-\bar{L}).
\end{aligned}
\nonumber
\end{equation}
Since  $\rho(\frac{r}{s_0+r\alpha}\bar{L} )<1$,  we have
$\left(I-(r/(s_0+r\alpha))\bar{L} \right) ^{-1}=\sum_{k=0}^{\infty}\left((r/(s_0+r\alpha))\bar{L} \right) ^{k}$.  This leads to
\begin{equation}
\begin{aligned}
(s_0I+rL)^{-1}L&=\frac{1}{s_0+r\alpha}\sum_{k=0}^{\infty}\left(\frac{r}{s_0+r\alpha}\bar{L} \right) ^{k}(\alpha I-\bar{L})
\\
&=\frac{\alpha}{s_0+r\alpha}I+\frac{\alpha}{s_0+r\alpha}\sum_{k=1}^{\infty}\left(\frac{r}{s_0+r\alpha}\bar{L} \right) ^{k}
-\frac{1}{r}\sum_{k=1}^{\infty}\left(\frac{r}{s_0+r\alpha}\bar{L} \right) ^{k}
\\
&=\frac{\alpha}{s_0+r\alpha}I-\frac{s_0}{r(s_0+r\alpha)} \sum_{k=1}^{\infty}\left(\frac{r}{s_0+r\alpha}\bar{L} \right) ^{k}
\\
&=\frac{s_0}{r(s_0+r\alpha)}\left(\frac{r\alpha}{s_0} I-\sum_{k=1}^{\infty}\left(\frac{r}{s_0+r\alpha}\bar{L} \right) ^{k} \right),
\end{aligned}
\nonumber
\end{equation}
where $\left((r/(s_0+r\alpha))\bar{L} \right) ^{k}\in \pi(\mathbb{R}_+^n),\,\forall k,$.
Furthermore, 
\begin{equation}
\rho\left( \sum_{k=1}^{\infty}\left(\frac{r}{s_0+r\alpha}\bar{L} \right) ^{k}\right) =\sum_{k=1}^{\infty}\left(\frac{r\alpha}{s_0+r\alpha} \right) ^{k}=\frac{r\alpha}{s_0},
\nonumber
\end{equation}
where the Perron-Frobenius theorem yields the first equality.  Hence,   $(s_0I+rL)^{-1}L$  is an M-matrix with row sums equal to zero, that is 
a Laplacian matrix \cite{chen2017protecting}.  In addition, it is irreducible.  We hence claim that all the entries of $\left(s_0I+(s_0I+rL)^{-1}L\right) ^{-1}$ are strictly positive.


For the case (iv),  we consider
the matrix $I-s_0(s_0^2I+(s_0r+1)L)^{-1}(s_0I+rL)$.  We have proved that
all entries of $s_0(s_0^2I+(s_0r+1)L)^{-1}(s_0I+rL)$ are strictly positive. On the other hand, one can easily see that $s_0(s_0^2I+(s_0r+1)L)^{-1}(s_0I+rL)=\left(I+(I+(r/s_0)L)^{-1}L\right) ^{-1}$ is a row stochastic matrix since $I+(I+(r/s_0)L)^{-1}L$ itself is row stochastic, which yields that 
\begin{equation}
\begin{aligned}
\left(I-s_0(s_0^2I+(s_0r+1)L)^{-1}(s_0I+rL)\right)_{i,i}&>0,\\
\left(I-s_0(s_0^2I+(s_0r+1)L)^{-1}(s_0I+rL)\right)_{i,j}&<0, \, \forall i\neq j.
\end{aligned}
\nonumber
\end{equation}

Consider the case $s_0=0$.  In light of Lemma \ref{lemma:Laplacian e_i}, in the case (ii), the associated matrix $\left[ -A,b;c^{\top}\!\!,0\right] $
has full  rank.    In the case (iii),  the matrix
 is not full rank,  but it can be easily verified that $d_0$ associated with $s_0=0$ always equals to zero. Therefore,  for the attack with $\forall b\in \Lambda_2~( \Lambda_1)$,  any defense $c$ located in the set $\Pi$ is (almost) successful.
\end{proof}

In   Theorem  \ref{theorem:secondt-order strongly connected multi-agent 2},  it holds that $\Lambda_1\cup \Lambda_2=\Lambda_{\mathbb{R}_{+}^{2n}}$. Therefore,
together with Theorem \ref{theorem:secondt-order strongly connected multi-agent 1},  for  both second-order multi-agent systems \eqref{eq:second-order multi-agent-closedloop 1} and \eqref{eq:second-order multi-agent-closedloop 2},
 all the potentially vulnerable positions for undetectable attacks have been considered.   At the same time,   an arbitrary single measurement in $\Pi=\lbrace f_i, i=1,\ldots,2n\rbrace$ is successful to prevent those attacks.

\section{Further Discussion on MIMO Case}  \label{sec:6}
In the preceding sections, we proposed  a single-sensor placement strategy to defend a single attack. This strategy is effective for 
single-input,  single-output (SISO) systems.  We now discuss its extension to MIMO systems.

Assume that multiple undetectable attacks enter the system
through $m$ channels. Specifically,  in the system \eqref{eq:mimolti}, let
 $B=[b_1,\ldots,b_m]\in \mathbb{R}^{n\times m}$,  and $C^{\top}=[c_1,\ldots,c_l] \in \mathbb{R}^{n\times l}$,  where $b_i,~c_j$ 
 are the $i$th and $j$th Euclidean coordinates of $\mathbb{R}^n$, respectively.
 According to
Definition \ref{definition:successful defense strategy}, for  the multiple attacks, the defense strategy will be  successful if and only if
\begin{equation}\label{eq:mimo rank condition}
\rank(C(s_0I-A)^{-1}B)=m
\end{equation}
for 
 each $s_0$ with $s_0\geq 0$.
 Apparently, for the extreme case such that $m=l=n$ and $B=C=I$, it yields that 
$$\rank(C(s_0I-A)^{-1}B)=\rank((s_0I-A)^{-1})=n.$$
That is, if all nodes and channels are placed by sensors, then even they all have been attacked, the defense strategy is still successful. Nevertheless, if $m\ll n$, together with the cost-efficient consideration, 
the complete sensor placement is indeed impractical.  Under this circumstance, we assume that $m=l<n$. The condition \eqref{eq:mimo rank condition} can reduce to
\begin{equation}\label{eq:mimo det condition}
\det(C(s_0I-A)^{-1}B)\neq 0,~~~\forall s_0\geq 0.
\end{equation}
 We use $\det$ to  represent the matrix determinant.
 Different from the SISO case concentrating on the entries of $(s_0I-A)^{-1}$,  for the MIMO case, it requires us to study the minors of $(s_0I-A)^{-1}$. The cut-down columns and rows are determined by the matrices $B$ and $C$.

 As an example, we design defense strategy against multiple attacks over the  first-order multi-agent system \eqref{eq: first-order multi-agent system} $\dot{x}(t)=-Lx(t)$ to demonstrate the highly increased complexity from the SISO scenario to the MIMO one.
From the condition \eqref{eq:mimo det condition}, it is necessary to examine 
$\det(C(s_0I+L)^{-1}B),~~~\forall s_0> 0.$
Denote by $\mathcal{I}_{B},~\mathcal{I}_C\subseteq \lbrace1,\ldots,n\rbrace$ the sets of integers with the cardinality $m$ indexed by the matrices $B$ and $C$. Then we use 
$(s_0I+L)^{-1}_{\mathcal{I}_{C} \mathcal{I}_B}$ to denote the submatrix of $(s_0I+L)^{-1}$ restricted to rows ordered by $C$ and columns ordered by $B$. From the Jacobi's complementary minor formula \cite{caracciolo2013algebraic}, it follows that  $\det((s_0I+L)^{-1}_{\mathcal{I}_{C}\mathcal{I}_B})\neq 0$ if and only if $\det((s_0I+L)_{\bar{\mathcal{I}}_{B}\bar{\mathcal{I}}_C})\neq 0$, where $\bar{\mathcal{I}}_{B}$ and $\bar{\mathcal{I}}_C$ are the complementary sets of $\mathcal{I}_{B}$ and $\mathcal{I}_C$. We add a fictitious vertex on the underlying digraph indexed by $0$, where all other vertices point to the vertex $0$ with the weight $s_0$ on each new edge. The enlarged Laplacian matrix is noted as $\hat{L}\in \mathbb{R}^{(n+1)\times (n+1)}$. Then it is straightforward that
$$
\det((s_0I+L)_{\bar{\mathcal{I}}_{B}\bar{\mathcal{I}}_C})=\det(\hat{L}_{\bar{\mathcal{I}}_{B}\bar{\mathcal{I}}_C}).
$$
Therefore, the target injection set and the sensor placement set
should contain all available sets $\mathcal{I}_{B}$ and $\mathcal{I}_C$ such that  $\det(\hat{L}_{\bar{\mathcal{I}}_{B}\bar{\mathcal{I}}_C})\neq 0,~\forall s_0> 0.$ 
They can be solved effectively based on the underlying digraph's structure in light of the \textit{all minors matrix tree theorem} \cite{chaiken1982combinatorial}.
Note that if the digraph is strongly connected and $m=1$, we can arrive at Corollary \ref{coro:first-order strongly connected multi-agent} directly from the \textit{matrix tree theorem}.  



 In the next step, we will continue to extend our
analysis to the multiple attacks over more general cone-invariant systems in a progressive manner. 

\section{Examples}  \label{sec:7}

In this section, we use two power network models as examples to illustrate our results, 
demonstrating how sensors may be placed to prevent undetectable attacks. We consider the decentralized voltage control schemes in interconnected buses over networks. Each bus can be controlled to regulate its voltage magnitude and phase-angle independently.   We refer to \cite{kundur1994power,anderson2008power} for the network models and the terminologies.

\begin{exmp}\label{ex:positive}
In this example, we  verify the result in Theorem \ref{theorem:sensor sets in cis} for positive systems. We start by considering the voltage magnitude dynamics of the power network that consists of a set of interconnected buses. Assume that the power network is outfitted with $n$ buses which
 are all inverter buses(microgrids). For bus $i$, its voltage magnitude  dynamics can be  modeled by a single integrator 
$$\tau_i\dot{V}_i(t)=u_i(t),$$
where $\tau_i>0$ is the inverter's time-constant. Consider the voltage quadratic droop controller \cite{simpson2016voltage} described by
\begin{equation}
u_i(t)=-\kappa_iV_i(t)(V_i(t)-V_i^{\ast}(t))-Q_i(t),
\nonumber
\end{equation}
where $\kappa_i>0$,  $V_i^{\ast}$, and $Q_i$ are the controller gain, the nominal voltage, and the reactive power injection at bus $i$, respectively. The reactive power injection $Q_i$ is given by:
$$
Q_i(t) = -\sum_{j\in \mathcal{N}_i}k_{ij} \cos(\theta_i(t)-\theta_j(t)),
$$
where $k_{ij}=V_iV_jb_{ij}$, $b_{ij}$ is the susceptance of the power line between buses $i$ and $j$, and $\theta_i$ is the phase angle of bus $i$. Since 
all the phase angles are close, under the standard decoupling approximation, it follows that $\cos(\theta_i-\theta_j)\approx 1$. Denoting $V=[V_1,\ldots,V_n]^{\top}$, $\tau=[\tau_1,\ldots,\tau_n]^{\top}$, and $\kappa=[\kappa_1,\ldots,\kappa_n]^{\top}$. After the Jacobian linearization around an equilibrium point $\bar{V}$,  
the  state-space form  \eqref{eq:cis} of  the power network's voltage magnitude  dynamics is given by the state $x=V-\bar{V}$ and the system matrix
$$A=-\mathrm{diag}(\bar{V})\mathrm{diag}(\tau)^{-1}\left( \mathrm{diag}(\kappa)+B_{\mathcal{G}} \right), $$
where  $B_{\mathcal{G}}$ is the susceptance matrix over the underlying topology.

\begin{figure}[!htb]
\begin{center}
\includegraphics[width=8.2cm, height=4.5cm]{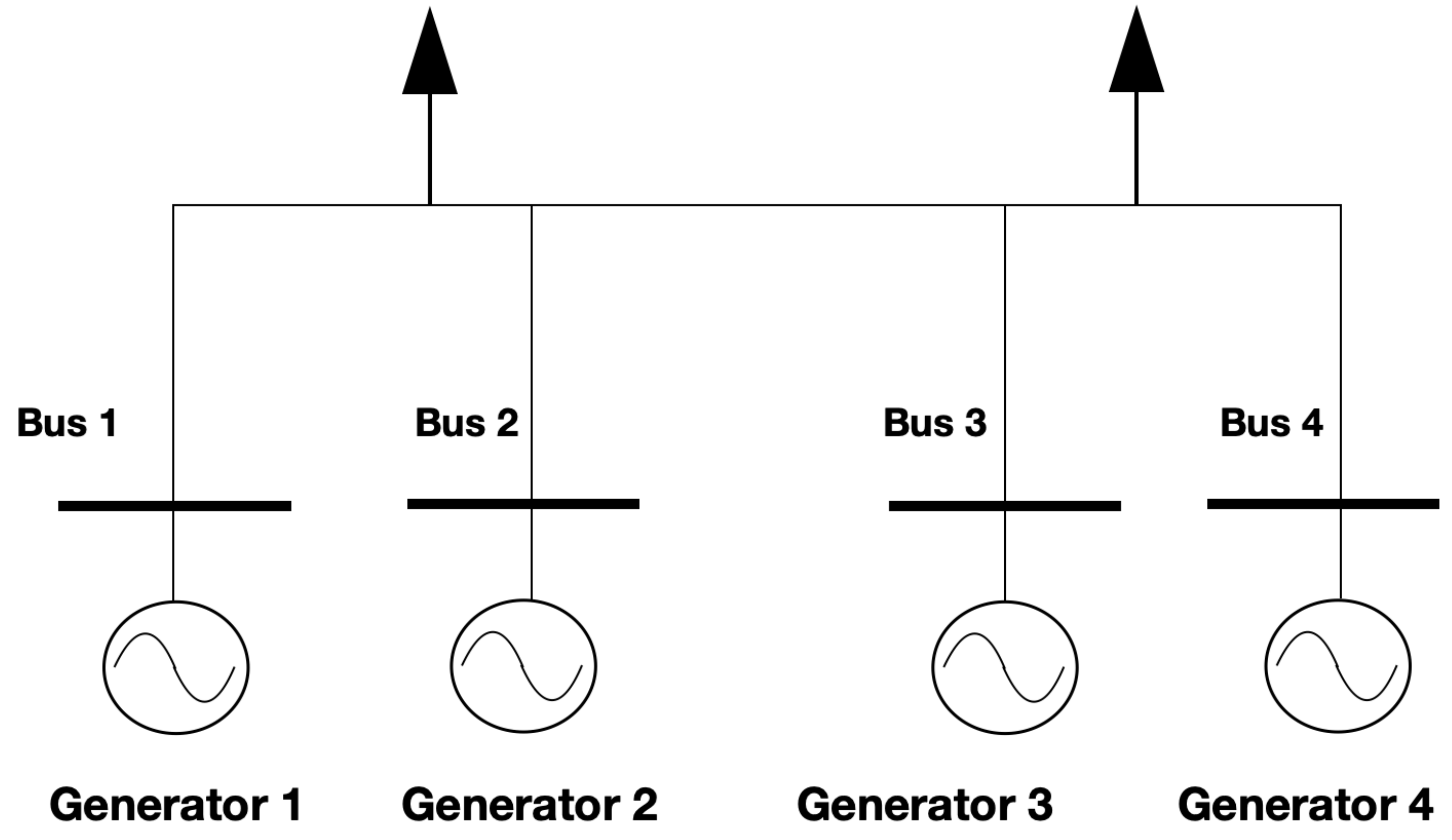}
\vspace*{-1.0mm}
\caption{The 4-bus power network with a line topology.}
\label{fig:4-bus}
\end{center}
\end{figure}

Consider the 4-bus power network with a line topology, as depicted in Fig \ref{fig:4-bus} \cite{simpson2016voltage}. It can be treated as the simplest model used in studies of the schematic of a parallel microgrid.  For simplicity, we assume that $\bar{V}_i=1$ and $\tau_i=1$ for all inverter buses. The controller gains $\kappa$ and the corresponding susceptance matrix $B_{\mathcal{G}}$ were  arbitrarily generated within reasonable ranges.
The voltage magnitude dynamics is characterized by
\begin{equation}\label{e:ex-positive system}
\begin{aligned}
A&=\begin{bmatrix}
-4.03&1.48&0&0\\1.48&-3.57&1.57&0\\0&1.57&-3.24&0.64\\0&0&0.64&-1.25
\end{bmatrix},\quad b=\begin{bmatrix}
1\\0\\0\\0
\end{bmatrix},\\
c^{\top}&=\,\begin{bmatrix}
\,\quad 0&\ \ \,\,\quad 0&\ \ \qquad1&\,\!\quad\quad0\ \,
\end{bmatrix},
\end{aligned}
\nonumber
\end{equation}
where $A$ is an irreducible and stable Metzler matrix, whilst $b$ and $c$ belong to $\Lambda_{\mathbb{R}_{+}^4}=\Pi=\lbrace e_i,i=1,\ldots,4\rbrace$. 
Since the unique real zero of the system is $-1.25$, one possible solution of \eqref{eq:rosenbrock} with $s_0=-1.25$ is $x(0)-\tilde{x}(0)=[
-5.5800\;-3.5596\;0\;8.7322]^{\top}$ and $d_0=10.2441$. This implies an attack signal is $d(t)=-10.2441e^{-1.25t}$.

In Fig. \ref{fig:attack+fake}, the solid lines indicate the trajectories initiated from  $x(0)=[
 12.5822$
 $10.0375\;
   13.4447\;14.7301]^{\top}$ with attack $d(t)$, while the dot lines are initiated from the possibly fraudulent initial state $\tilde{x}(0)=[
  18.1621\;13.5972\;13.4447\; 5.9979]^{\top}$. Note that the outputs of these two cases  coincide always, which shows the undetectability of the attack $d(t)$. In Fig. \ref{fig:attack+normal}, there are the trajectories initiated from $x(0)$ with (solid lines) and without (dash-dot lines) the attack $d(t)$. The attack $d(t)$ cannot affect the asymptotic stability of the system. In these figures, $x(t)$, $\tilde{x}(t)$, and $\hat{x}(t)$ represent the normal state trajectories, the fraudulent state trajectories, and the attacked state trajectories, respectively. 
\begin{figure}[!htb]
\begin{center}
\includegraphics[width=8.7cm, height=6.7cm]{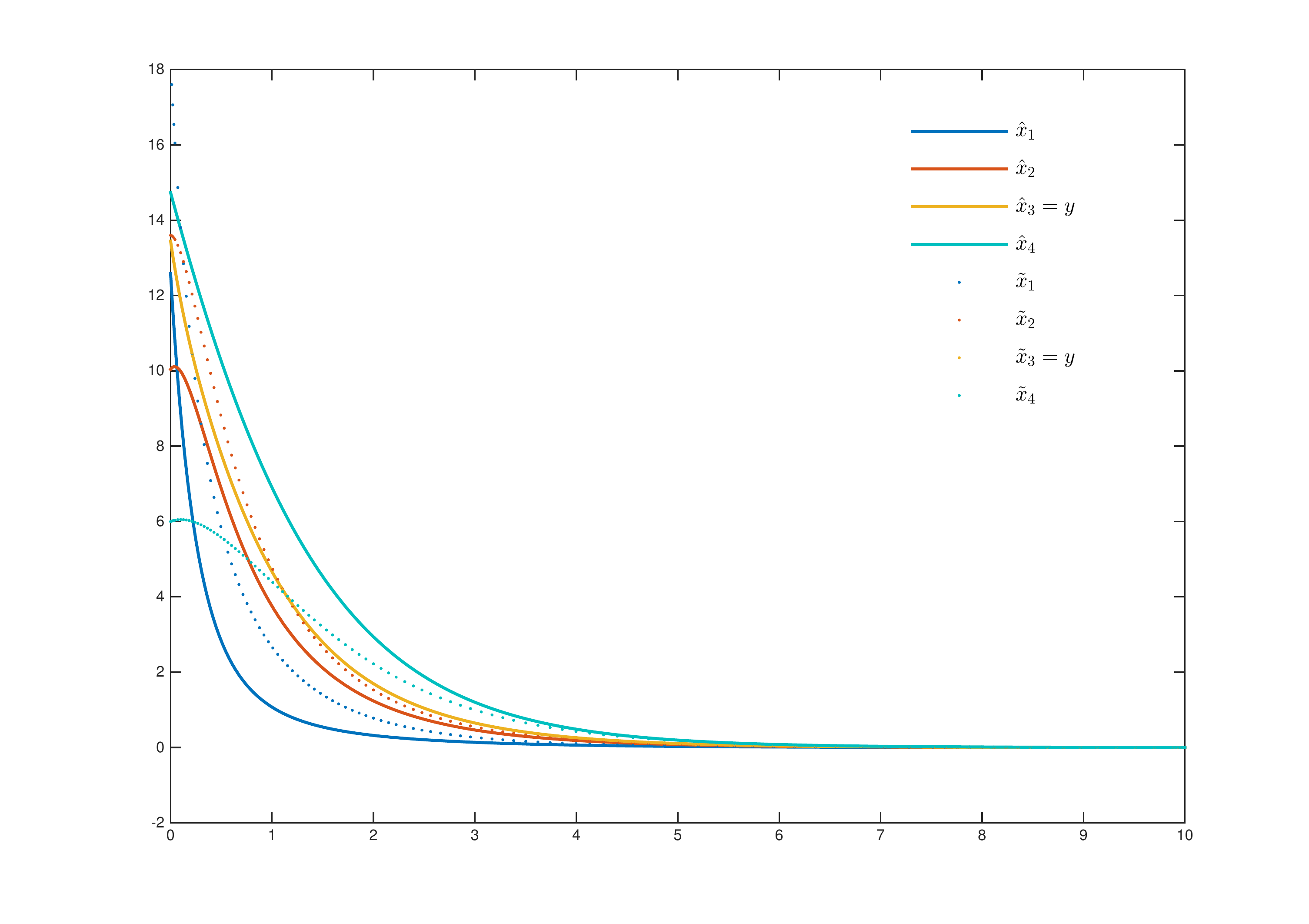}
\vspace*{-1.0mm}
\caption{The undetectability of the attack $d(t)$ on the 4-bus power network defined in Example \ref{ex:positive}.}
\label{fig:attack+fake}
\end{center}
\end{figure}
\begin{figure}[!htb]
\begin{center}
\includegraphics[width=8.7cm, height=6.7cm]{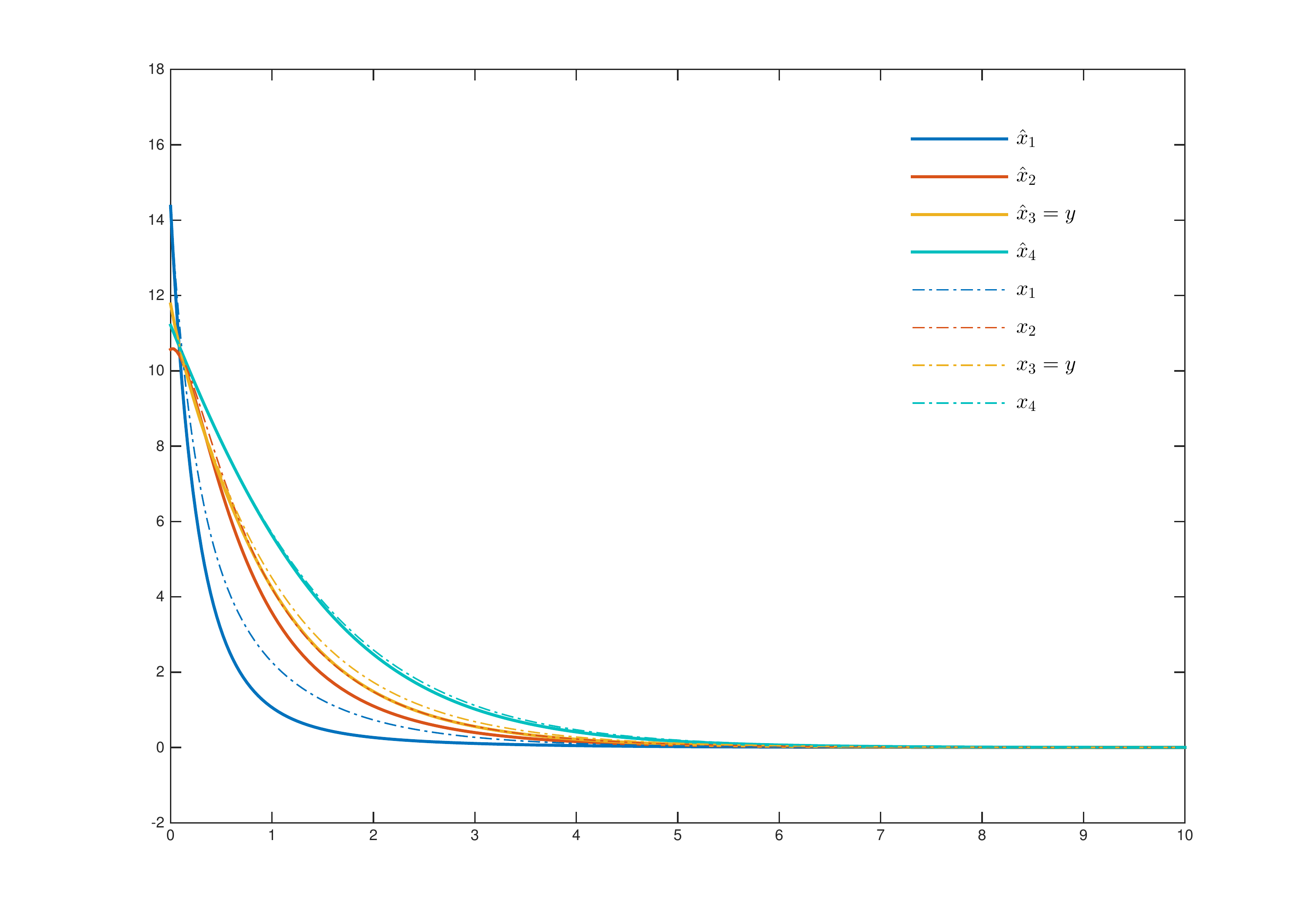}
\vspace*{-1.0mm}
\caption{The influence of the attack $d(t)$ on the 4-bus power network defined in Example \ref{ex:positive}.}
\label{fig:attack+normal}
\end{center}
\end{figure}
\end{exmp}

\begin{exmp} \label{e:ex-second order}
In this example, the result in Theorem \ref{theorem:secondt-order strongly connected multi-agent 1} 
is verified. We turn to consider  the phase-angle dynamics of power network with 9 buses, as depicted in Fig. \ref{fig:9-bus} \cite{anderson2008power}.
\begin{figure}[!htb]
\begin{center}
\includegraphics[width=8.7cm, height=6.2cm]{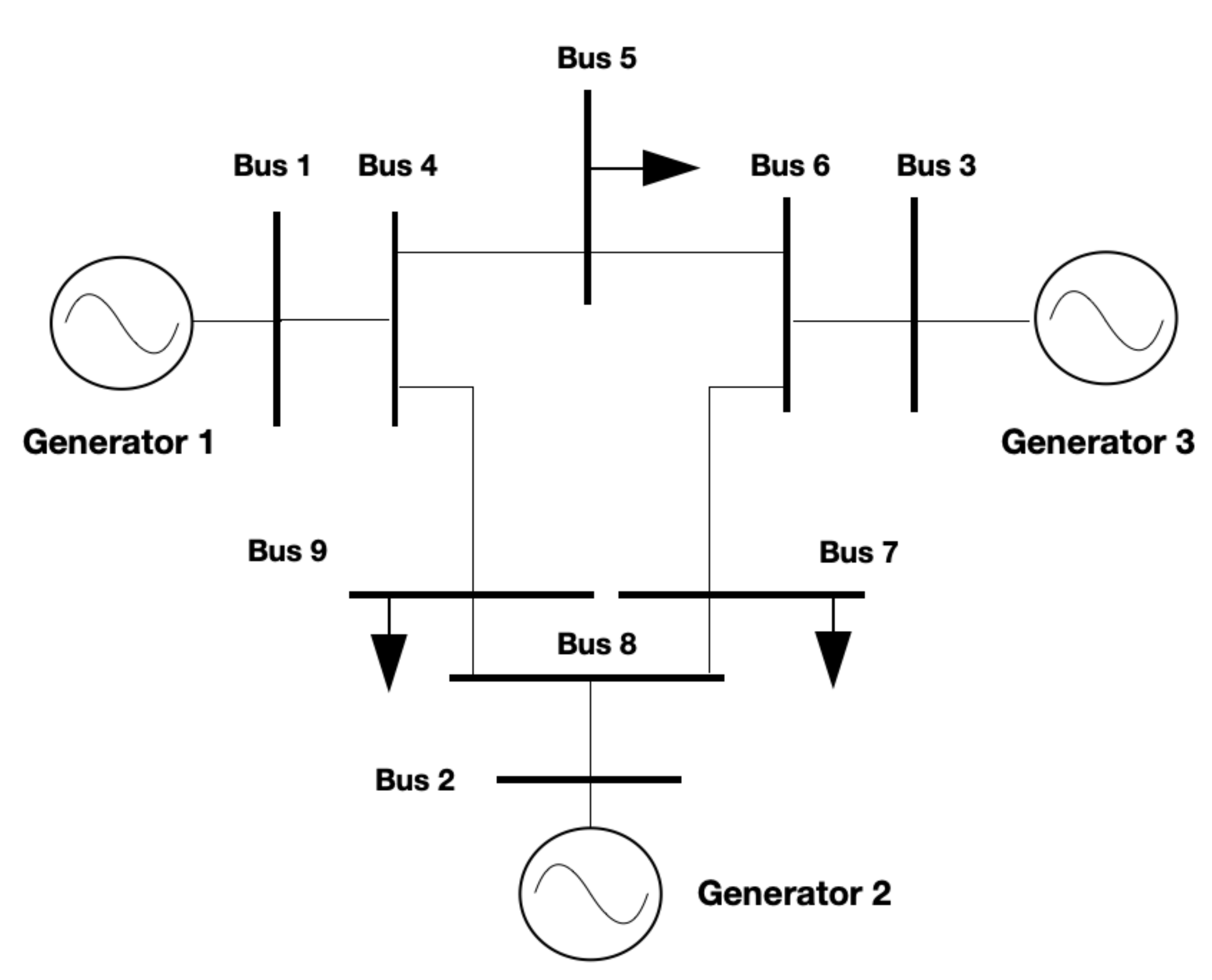}
\vspace*{-1.0mm}
\caption{The 9-bus power network.}
\label{fig:9-bus}
\end{center}
\end{figure}
Assume that all buses are synchronous cases.  The phase-angle dynamics of bus $i$ can be described by
the so-called swing equation:
\begin{equation}
m_i\ddot{\theta}_i(t)+d_i\dot{\theta}_i(t)-P_{mi}(t)+\sum_{j\in \mathcal{N}_i}P_{ij}(t)=0,
\nonumber
\end{equation}
where $\theta_i$ is the phase angle of bus $i$, $m_i$ and $d_i$ are the inertia
and damping coefficients, respectively, $P_{mi}$ is the mechanical input
power, and $P_{ij}$ is the active power flow from bus $i$ to $j$. The active power flow $P_{ij}$ is given by:
$$
P_{ij}(t) = k_{ij} \sin(\theta_i(t)-\theta_j(t)),
$$
where $k_{ij}=V_iV_jb_{ij}$, $V_i$ is the voltage magnitude of bus $i$, and $b_{ij}$ is the susceptance of the power line between buses $i$ and $j$.
Consider the case when the phase angles are close, after the linearization, the phase-angle dynamics of bus $i$ can be rewritten as: 
\begin{equation}
m_i\ddot{\theta}_i(t)+d_i\dot{\theta}_i(t)-P_{mi}(t)+\sum_{j\in \mathcal{N}_i}k_{ij} (\theta_i(t)-\theta_j(t))=0.
\nonumber
\end{equation}
The 9-bus power network's topological parameters are noted in Fig. \ref{fig:9-bus}. For simplicity, we normalized
the voltage magnitudes $V_{i}=1,\, \forall i$, while the rest dynamic coefficients of the buses were arbitrarily taken from reasonable values.
The phase-angle dynamics of the power network  can be rewritten  
in the state-space form \eqref{eq:second-order multi-agent-closedloop 1} with the size $A\in \mathbb{R}^{18\times 18}$. The system's state is denoted by $x=[\theta_1,\ldots,\theta_9,\dot{\theta}_1,\ldots,\dot{\theta}_9]^{\top}$.
The attacker's target is specified as $x_1$. We choose to measure $x_{10}$ as the defense strategy, which can be 
readily available through phase measurement units (PMU). 
The real zeros 
of the system are $0$ and $-0.1123$. 
For $s_0=0$, the solution to \eqref{eq:rosenbrock} satisfies $d_0=0$. 
For $s_0=-0.1123$, one solution to the equation \eqref{eq:rosenbrock} is $x(0)-\tilde{x}(0)=[
6.7189\;7.6055\;6.4301\;6.7189\;5.6482\;6.1816\;6.0540\;6.8980\\
\;7.1553\;0.0000\;-0.8538\;-0.7218\;-0.7542\;-0.6340\\-0.6939
\;-0.6796\;-0.7743-0.8032]^{\top}$ and $d_0=0.7542$, which implies the attack signal $d(t)=-0.7542e^{-0.1123t}$. 

For the clarity, only the behaviors of buses 1,3, and 5 are recorded in figures.
The solid lines in Fig. \ref{fig:attack+fake+9bus} are the trajectories initiated from $\left\lbrace x_i(0), i=1,3,5,10,12,14\right\rbrace$ with attack $d(t)$, while the dot lines are initiated from $ \left\lbrace \tilde{x}_i(0), i=1,3,5,10,12,14\right\rbrace$. Note that the outputs of these two cases (the cyan line) are exactly the same and the undetectability is confirmed. In Fig. \ref{fig:attack+normal+9bus}, there are the trajectories initiated from $x_{i}(0)$ with (solid lines) and without (dash lines) the attack $d(t)$. The 9-bus power system can still achieve synchronous with the attack $d(t)$ while the final synchronous points of $\left\lbrace x_j(t), j=1,3,5\right\rbrace $ are deflected.
\begin{figure}[!htb]
\begin{center}
\includegraphics[width=8.7cm, height=6.7cm]{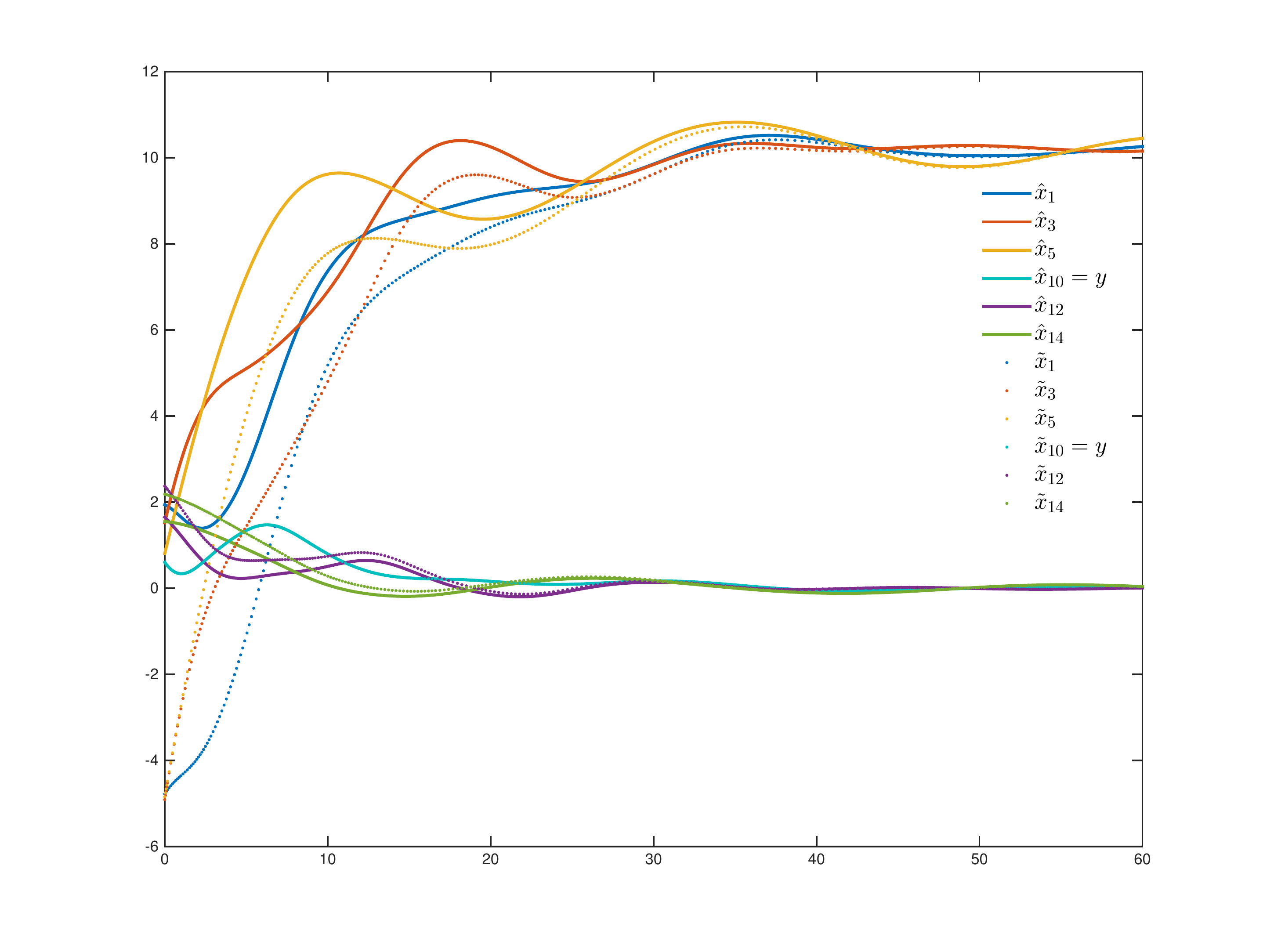}
\vspace*{-1.0mm}
\caption{The undetectability of the attack $d(t)$ on the 9-bus power network defined in Example \ref{e:ex-second order}.}
\label{fig:attack+fake+9bus}
\end{center}
\end{figure}

\begin{figure}[!htb]
\begin{center}
\includegraphics[width=8.7cm, height=6.7cm]{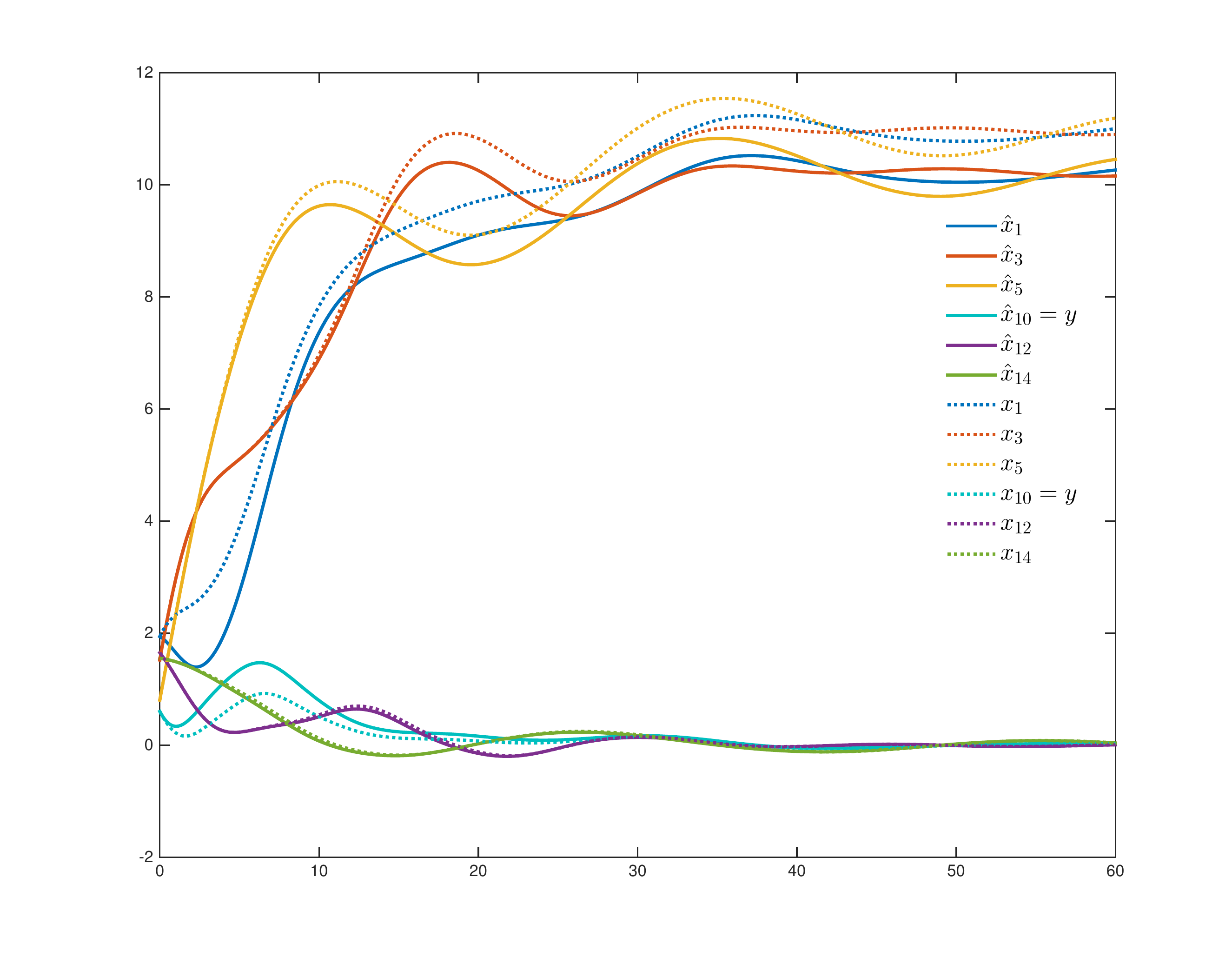}
\vspace*{-1.0mm}
\caption{The influence of the attack $d(t)$ on the 9-bus power network defined in Example \ref{e:ex-second order}.}
\label{fig:attack+normal+9bus}
\end{center}
\end{figure}
\end{exmp}

\section{Conclusion}  \label{sec:8}
An explicit and efficient detection and defense framework for cone-invariant systems
and multi-agent systems against undetectable attacks has been
developed in this paper. We have shown that for an undetectable attacks getting into the system through any position in $\Lambda_{\mathcal{K}}$, any defense by placing one sensor in the union $\Pi$ is successful or almost successful. The sets $\Lambda_{\mathcal{K}}$ and $\Pi$ can be characterized by the intersection of the canonical basis vectors and the cone $\mathcal{K}$ or its dual cone $\mathcal{K}^{\ast}$, which are easily computable from a geometric perspective. 
For multi-agent systems with single- and double-integrator dynamics,   we show that  if the associated digraph is strongly connected,
the corresponding sets $\Lambda_{\mathcal{K}}$ and $\Pi$ will both contain all nodes in the  digraph.

\bibliographystyle{IEEEtran}        
\bibliography{Cite}           

\end{document}